\documentclass[aip,jmp,preprint]{revtex4-1}%
\usepackage{amsfonts}
\usepackage{amsmath}
\usepackage{amssymb}
\usepackage{graphicx}
\usepackage[colorlinks=true,linkcolor=blue,citecolor=red,plainpages=false,pdfpagelabels]%
{hyperref}
\usepackage{fullpage}%
\providecommand{\U}[1]{\protect\rule{.1in}{.1in}}
\newtheorem{theorem}{Theorem}

\newtheorem{corollary}[theorem]{Corollary}

\newtheorem{lemma}[theorem]{Lemma}

\newtheorem{proposition}[theorem]{Proposition}
\newtheorem{remark}[theorem]{Remark}

\newenvironment{proof}[1][Proof]{\noindent\textbf{#1.} }{\ \rule{0.5em}{0.5em}}
\begin{document}

\title{\textbf{R\'{e}nyi relative entropies of quantum Gaussian states}}

\author{Kaushik P. Seshadreesan}
\affiliation{Max Planck Institute for the Science of Light,
Staudtstr.~2, 91058 Erlangen, Germany}
%
\author{Ludovico Lami}
\affiliation{F\'isica Te\`orica: Informaci\'o i Fen\`omens Qu\`antics,
Departament de F\'isica, Universitat Aut\`onoma de Barcelona, ES-08193
Bellaterra (Barcelona), Spain}
\author{Mark M. Wilde}
\affiliation{Hearne Institute for Theoretical Physics, Department of
Physics and Astronomy, Center for Computation and Technology, Louisiana State
University, Baton Rouge, Louisiana 70803, USA}

\begin{abstract}
The quantum R\'{e}nyi relative entropies play a prominent role in quantum
information theory, finding applications in characterizing error exponents and
strong converse exponents for quantum hypothesis testing and quantum
communication theory. On a different thread, quantum Gaussian states have been
intensely investigated theoretically, motivated by the fact that they are more
readily accessible in the laboratory than are other, more exotic quantum
states. In this paper, we derive formulas for the quantum R\'{e}nyi relative
entropies of quantum Gaussian states. We consider both the traditional (Petz)
R\'{e}nyi relative entropy as well as the more recent sandwiched R\'{e}nyi
relative entropy, finding formulas that are expressed solely in terms of the
mean vectors and covariance matrices of the underlying quantum Gaussian
states. Our development handles the hitherto elusive case for the
Petz--R\'{e}nyi relative entropy when the R\'{e}nyi parameter is larger than
one. Finally, we also derive a formula for the max-relative entropy of two
quantum Gaussian states, and we discuss some applications of the formulas
derived here.

\end{abstract}

\maketitle

\section{Introduction}

Motivated by the mathematical foundations of entropy, R\'enyi \cite{R61} defined the
following $\alpha$-dependent relative entropy as a function of two probability
distributions $p$ and $q$:%
\begin{equation}
D_{\alpha}(p\Vert q)\equiv \frac{1}{\alpha-1}\ln\!\left(  \sum_{x}p(x)^{\alpha
}q(x)^{1-\alpha}\right)  ,
\end{equation}
where $\alpha\in(0,1)\cup(1,\infty)$, and for $\alpha\in\{0,1,\infty\}$, the
quantity is defined in the limit. An important special case is the limit
$\alpha\rightarrow1$, for which the quantity converges to the relative entropy
$D(p\Vert q)$ \cite{kullback1951}:
\begin{equation}
\lim_{\alpha\rightarrow1}D_{\alpha}(p\Vert q)=D(p\Vert q)\equiv\sum_{x}%
p(x)\ln\!\left(  \frac{p(x)}{q(x)}\right)  .
\end{equation}
Since R\'enyi's work \cite{R61}, the quantity $D_{\alpha}(p\Vert q)$ has become known as the R\'enyi relative
entropy and has played an important role in hypothesis testing and information
theory \cite{C95,vEH14}. Most prominently, the R\'enyi relative entropy has
found operational interpretations in these contexts in terms of error
exponents or strong converse exponents, which respectively characterize the
exponential rate at which error probabilities decay to zero or increase to one
for a given information-processing task.

Towards the goal of developing the quantum generalization of the
aforementioned fields, several researchers have defined quantum extensions of
the R\'{e}nyi relative entropy \cite{P86,MDSFT13,WWY14}. Interestingly, in the
quantum case, there are several ways to go about this, due to the
non-commutativity of quantum states. A first way of generalizing the R\'{e}nyi
relative entropy was put forward in Ref.~\onlinecite{P86}, where the following quantity
was defined for two density operators $\rho$ and $\sigma$:%
\begin{equation}
D_{\alpha}(\rho\Vert\sigma)\equiv\frac{1}{\alpha-1}\ln\operatorname{Tr}%
\{\rho^{\alpha}\sigma^{1-\alpha}\},
\end{equation}
with $\alpha\in(0,1)\cup(1,\infty)$. It has since become known as the
Petz--R\'{e}nyi relative entropy and has the following limits:%
\begin{align}
\lim_{\alpha\rightarrow0}D_{\alpha}(\rho\Vert\sigma) &  =D_{0}(\rho\Vert
\sigma)=-\ln\operatorname{Tr}\{\Pi_{\rho}\sigma\},\\
\lim_{\alpha\rightarrow1}D_{\alpha}(\rho\Vert\sigma) &  =D(\rho\Vert
\sigma)\equiv \operatorname{Tr}\{\rho\left[  \ln\rho-\ln\sigma\right]  \},
\end{align}
where $\Pi_{\rho}$ is the projection onto the support of $\rho$ and
$D(\rho\Vert\sigma)$ is the quantum relative entropy \cite{U62,Lindblad1973}.
More recently, a second way of generalizing the R\'{e}nyi relative entropy was
put forward \cite{MDSFT13,WWY14}:
\begin{equation}
\widetilde{D}_{\alpha}(\rho\Vert\sigma)\equiv\frac{1}{\alpha-1}\ln
\operatorname{Tr}\left\{  \left(  \sigma^{\frac{1-\alpha}{2\alpha}}\rho
\sigma^{\frac{1-\alpha}{2\alpha}}\right)  ^{\alpha}\right\}
.\label{eq:sandwich-RE}%
\end{equation}
This quantity is known as the sandwiched R\'{e}nyi relative entropy, due to
the operator sandwich in \eqref{eq:sandwich-RE}, and it has the following
limits \cite{MDSFT13,WWY14}:%
\begin{align}
\lim_{\alpha\rightarrow1}\widetilde{D}_{\alpha}(\rho\Vert\sigma) &
=D(\rho\Vert\sigma),\\
\lim_{\alpha\rightarrow\infty}\widetilde{D}_{\alpha}(\rho\Vert\sigma) &
=D_{\max}(\rho\Vert\sigma)\equiv\inf\left\{  \lambda\in\mathbb{R}:\rho\leq
e^{\lambda}\sigma\right\}  ,
\end{align}
where $D_{\max}$ denotes the max-relative entropy \cite{D09}. Both the
Petz--R\'{e}nyi relative entropy and the sandwiched R\'{e}nyi relative entropy
have found widespread application in quantum hypothesis testing and quantum
communication theory
\cite{ON99,ON00,OH04,N06,PhysRevA.76.062301,ANSV08,KW09,MH11,SW12,WWY14,MO13,GW13,CMW14,HT16,TWW14,DingW15,LWD16,WTB16}%
. Particular to the quantum case, all evidence to date indicates that the
Petz--R\'{e}nyi relative entropy is the appropriate quantity to employ in the
error exponent regime and the sandwiched R\'{e}nyi relative entropy in the
strong converse regime.

Along a different line, the theory of Gaussian quantum information has been
intensely investigated and developed \cite{Olivares2012,adesso14,S17}, the
main motivation behind it being that bosonic Gaussian states and evolutions
are more accessible in the laboratory than are their non-Gaussian
counterparts. These states and evolutions play a prominent role in quantum
optics, but they can also describe the physics of particular superconducting
degrees of freedom, trapped ions, and atomic ensembles \cite{S17}. Similar to
the classical case, a quantum Gaussian state of $n$ modes is uniquely
characterized by a mean vector (first moments)\ and a covariance matrix
(second moments). Furthermore, a quantum Gaussian channel is defined to take
Gaussian states to Gaussian states, and as such, one can uniquely characterize
a quantum Gaussian channel by its action on the mean vector and covariance
matrix of an input Gaussian state \cite{CEGH08}. These simple
characterizations are helpful for theoretical manipulations:\ even though
Gaussian states are density operators acting on infinite-dimensional,
separable Hilbert spaces, it often suffices to manipulate their finite-dimensional
mean vectors and covariance matrices. A typical goal is to express
information-theoretic functions of Gaussian states solely in terms of their
mean vectors and covariance matrices, so that these functions can be easily
evaluated numerically or analytically.

With these two threads in mind, the contribution of the present paper lies at
the convergence of them. That is, in this paper, we establish formulas for the
Petz--R\'{e}nyi relative entropy and the sandwiched R\'{e}nyi relative entropy
of any two quantum Gaussian states. As desired, these formulas are expressed
solely in terms of the mean vectors and covariance matrices of the two states.
The most direct consequence of our formulas is in quantum state
discrimination, such that it is now possible to characterize error exponents
and strong converse exponents in terms of our formulas. We discuss the
application to quantum state discrimination in Section~\ref{sec:state-disc}.
Given the many applications of quantum R\'{e}nyi relative entropies, we expect
there to be further applications of the formulas provided here.

Two special cases of our formulas have already appeared in the literature, and
so it is pertinent to recall these developments now. To see the first one, we should note that
the following limit holds for the sandwiched R\'enyi relative entropy:%
\begin{equation}
\lim_{\alpha\rightarrow\frac{1}{2}}\widetilde{D}_{\alpha}(\rho\Vert
\sigma)=-\log F(\rho,\sigma),
\end{equation}
where $F(\rho,\sigma)$ denotes the well known quantum fidelity \cite{U76}:%
\begin{equation}
F(\rho,\sigma)=\left[  \operatorname{Tr}\left\{  \sqrt{\sqrt{\sigma}\rho
\sqrt{\sigma}}\right\}  \right]  ^{2}.
\end{equation}
Due to the significance of fidelity in quantum information theory, a number of
works have already devised formulas for the fidelity of quantum Gaussian
states. The authors of Refs.~\onlinecite{PS00,WKO00}\ determined a general formula for the
fidelity of two zero-mean Gaussian states.
The authors of Ref.~\onlinecite{WKO00} found the first general formula for the fidelity of two zero-mean Gaussian states in terms of their Hamiltonian matrices, using the tools of Ref.~\onlinecite{Balian1969}. In Ref.~\onlinecite{PS00}, the determination of the characteristic function of a Gaussian state sandwiched by the square root of another Gaussian state led to a simpler expression involving only the corresponding covariance matrices, again for the zero-mean case. A number of special cases have also been considered in several contributions \cite{T96,S98,PS98,MMS03,NC05,OPA06,MM08,Marian2008,BJORK20104440,GBMM10}. Some years after these developments, a general formula for the fidelity between two-mode Gaussian states was derived in Ref.~\onlinecite{MM12} (see also the review in Ref.~\onlinecite{Olivares2012}). In Ref.~\onlinecite{MM12}, an expression for the $n$-mode case was also given, which can be evaluated numerically. More explicit formulas to deal with this latter case were found recently in Ref.~\onlinecite{BBP15}. See Ref.~\onlinecite{PhysRevA.93.052330} for further developments.

In addition to the fidelity of Gaussian states, researchers have also
investigated the Petz--R\'{e}nyi relative entropy of Gaussian states
exclusively for the case when $\alpha\in(0,1)$. The authors of Ref.~\onlinecite{CMMAB08}
contributed a formula for the Petz--R\'{e}nyi relative entropy for the case of
single-mode Gaussian states, and this approach was generalized to the $n$-mode
case in Ref.~\onlinecite{PL08}. It is worthwhile to note that these authors were
interested in the symmetric error exponent of quantum hypothesis testing and
that the Petz--R\'{e}nyi relative entropy arises naturally in this context. The particular case of $\alpha = 1/2$ for the Petz--R\'enyi relative entropy was considered in Ref.~\onlinecite[Lemma~2]{H72} and Ref.~\onlinecite{MM15}, wherein a compact formula was given for this case.

In light of these prior works, the main contribution of our paper can be
understood as a general formula for the Petz--R\'{e}nyi relative entropy for
$\alpha\in(1,\infty)$ and for the sandwiched R\'{e}nyi relative entropy for
$\alpha\in(0,1)\cup(1,\infty)$. Of especial interest is the hitherto elusive
case for the Petz--R\'{e}nyi relative entropy when $\alpha\in(1,\infty)$.
Additionally, we derive an alternate expression for the Petz--R\'{e}nyi
relative entropy for $\alpha\in(0,1)$. We find that these formulas simplify
significantly when $\alpha=2$, and we also devote a section to the derivation
of a formula for the max-relative entropy of quantum Gaussian states.
Specifically, our main results are as follows:

\begin{enumerate}
\item Theorem~\ref{thm:petz-renyi-alpha<1}\ gives a formula for the
Petz--R\'{e}nyi relative entropy for $\alpha\in(0,1)$.

\item Theorem~\ref{prop:petz-renyi-alpha>1}\ gives a formula for the
Petz--R\'{e}nyi relative entropy for $\alpha\in(1,\infty)$.

\item Theorem~\ref{thm: Sandwiched RRE a<1}\ gives a formula for the
sandwiched R\'{e}nyi relative entropy for $\alpha\in(0,1)$.

\item Theorem~\ref{thm:sandwiched-alpha>1}\ gives a formula for the sandwiched
R\'{e}nyi relative entropy for $\alpha\in(1,\infty)$.

\item Theorem~\ref{thm:max-rel-ent}\ gives a formula for the max-relative entropy.
\end{enumerate}

\noindent The main tools that we use to derive these formulas are those that
were developed to derive the fidelity formula
\cite{Balian1969,PS00,WKO00,MM12,BBP15}. Given the prominence of both the
R\'{e}nyi relative entropies and quantum Gaussian states in quantum
information theory, we expect that the formulas derived here will find
application in a variety of avenues in quantum information and other areas of physics.

We organize our paper as follows. In Section~\ref{sec:prelim}, we review some
basics of quantum Gaussian states that are needed for the remainder of the
paper. Section~\ref{sec:computations-guassian} is devoted to recalling and
proving analytic forms for several mappings of quantum Gaussian states.\ After
this preparatory material, Section~\ref{sec:Petz-Renyi} offers a derivation of
the Petz--R\'{e}nyi relative entropy of two quantum Gaussian states, first for $\alpha\in(0,1)$ and then for
$\alpha\in(1,\infty)$. Section~\ref{sec:sandwiched-Renyi} gives a derivation
of the sandwiched R\'{e}nyi relative entropy
of two quantum Gaussian states, for $\alpha\in(0,1)$ and then for
$\alpha\in(1,\infty)$. In Section~\ref{sec:max-rel-ent}, we derive a formula
for the max-relative entropy of two quantum Gaussian states. We then discuss
applications of our results in Section~\ref{sec:apps}, and we conclude
in Section~\ref{sec:conclusion} with a summary and some open questions.

\section{Preliminaries on quantum Gaussian states}

\label{sec:prelim}

We begin with a brief review of quantum Gaussian states and point the reader
to Ref.~\onlinecite{S17} for more background. Our development applies to $n$-mode
Gaussian states, where $n$ is some fixed positive integer. Let $\hat{x}_{j}$
denote each quadrature operator ($2n$ of them for an $n$-mode state), and let
\begin{equation}
\hat{x}\equiv\left[  \hat{q}_{1} ,\ldots,\hat{q}_{n},\hat{p}_{1},\ldots
,\hat{p}_{n}\right]  \equiv\left[  \hat{x}_{1},\ldots,\hat{x}_{2n}\right]
\end{equation}
denote the vector of quadrature operators, so that the first~$n$ entries
correspond to position-quadrature operators and the last~$n$ to
momentum-quadrature operators. The quadrature operators satisfy the following
commutation relations:
\begin{equation}
\label{eq:symplectic-form}\left[  \hat{x}_{j},\hat{x}_{k}\right]
=i\Omega_{j,k},\quad\mathrm{where}\quad\Omega=
\begin{bmatrix}
0 & 1\\
-1 & 0
\end{bmatrix}
\otimes I_{n},
\end{equation}
and $I_{n}$ is the $n\times n$ identity matrix. Note that $\Omega^{T}=-\Omega$
and the matrix $i\Omega$ is involutory, (i.e., $\left(  i\Omega\right)
\left(  i\Omega\right)  =I$) facts that we use repeatedly in what follows.

A faithful Gaussian state $\rho$ of $n$ modes can be written as \cite{S17}%
\begin{align}
\rho &  =\frac{1}{Z_{\rho}}\exp\left[  -\frac{1}{2}(\hat{x}-s_{\rho}%
)^{T}H_{\rho}(\hat{x}-s_{\rho})\right]  ,\label{eq: generic Gaussian}\\
Z_{\rho}  &  \equiv\sqrt{\det(\left[  V_{\rho}+i\Omega\right]  /2)},
\end{align}
where $H_{\rho}$ is a $2n\times2n$ positive-definite real Hamiltonian matrix,
$s_{\rho}\in\mathbb{R}^{2n}$ is the mean vector, defined as $s_{\rho
}=\left\langle \hat{x}\right\rangle _{\rho}=\operatorname{Tr}\{\hat{x}\rho\}$,
and $V_{\rho}$ is the symmetric covariance matrix, whose entries are defined
as%
\begin{equation}
\left[  V_{\rho}\right]  _{j,k}=\left\langle \left\{  \hat{x}_{j}-s^j_{\rho
},\hat{x}_{k}-s^k_{\rho}\right\}  \right\rangle _{\rho}.
\end{equation}
The matrices $V_{\rho}$ and $H_{\rho}$ are related by
\cite{PhysRevA.71.062320,K06,Holevo2011,H11,H12}%
\begin{align}
H_{\rho}  &  =2i\Omega\operatorname{arcoth}(V_{\rho}i\Omega
),\label{eq:cov-to-gibbs}\\
V_{\rho}  &  =\coth(i\Omega H_{\rho}/2)i\Omega, \label{eq:gibbs-to-cov}%
\end{align}
where%
\begin{align}
\coth(x)  &  =\frac{e^{x}+e^{-x}}{e^{x}-e^{-x}}=\frac{e^{2x}+1}{e^{2x}-1},\\
\operatorname{arcoth}(x)  &  =\frac{1}{2}\ln\!\left(  \frac{x+1}{x-1}\right)
.
\end{align}
These relationships imply  for finite $H$ that%
\begin{equation}
H_{\rho}>0\qquad\Longleftrightarrow\qquad V_{\rho}+i\Omega>0.
\label{eq:H-PD-V-legit}%
\end{equation}
Note that the condition $V_{\rho}+i\Omega\geq 0$, if the state is not necessarily faithful, leaves open the possibility that $H$ diverges.
We say that $V_{\rho}$ is a legitimate covariance matrix if it satisfies the
following uncertainty principle \cite{PhysRevA.49.1567}:
\begin{equation}
V_{\rho}+i\Omega\geq0,
\end{equation}
and we note that, by a transpose, this is equivalent to $V_{\rho}-i\Omega
\geq0$.

Alternatively, given a positive-definite real matrix $H_{\rho}$, we have that
\cite{MS78,FPJ85}
\begin{equation}
\operatorname{Tr}\left\{  \exp\left[  -\frac{1}{2}\hat{x}^{T}H_{\rho}\hat
{x}\right]  \right\}  =\sqrt{\det(\left[  V_{\rho}+i\Omega\right]  /2)},
\label{eq:norm-of-Gaussian-form}%
\end{equation}
where $V_{\rho}=\coth(i\Omega H_{\rho}/2)i\Omega$.
We also consider operators of the form $\exp\left[  -\frac{1}%
{2}\hat{x}^{T}H\hat{x}\right]  $, in which $H$ is a symmetric matrix
with complex entries. In this case, we can still exploit the functional
relationships in \eqref{eq:cov-to-gibbs} and \eqref{eq:gibbs-to-cov}. Note
that $H$ is symmetric if and only if $V$ is symmetric, one direction of which
can be seen from the following:%
\begin{align}
V^{T}  &  =\left[  \coth(i\Omega H/2)i\Omega\right]  ^{T}\\
&  =-i\Omega\coth(-Hi\Omega/2)\\
&  =i\Omega\coth(Hi\Omega/2)\left(  i\Omega\right)  \left(  i\Omega\right) \\
&  =\coth(\left(  i\Omega\right)  Hi\Omega\left(  i\Omega\right)  /2)\left(
i\Omega\right) \\
&  =\coth(i\Omega H/2)i\Omega\\
&  =V,
\end{align}
with the other implication following similarly. In the above, we used the fact
that $\coth$ is an odd function and the functional analytic relation
$Mf(L)M^{-1}=f(MLM^{-1})$.

A $2n\times2n$ real matrix $S$ is symplectic if it preserves the symplectic form:
$S\Omega S^{T}=\Omega$. According to Williamson's theorem \cite{W36}, there is
a diagonalization of the covariance matrix $V_{\rho}$ of the form%
\begin{equation}
V_{\rho}=S_{\rho}\left(  D_{\rho}\oplus D_{\rho}\right)    S_{\rho
}  ^{T},
\end{equation}
where $S_{\rho}$ is a symplectic matrix and $D_{\rho}\equiv\operatorname{diag}%
(\nu_{1},\ldots,\nu_{n})$ is a diagonal matrix of symplectic eigenvalues, such
that $\nu_{i}\geq1$ for all $i\in\left\{  1,\ldots,n\right\}  $. A quantum
Gaussian state is faithful if all of its symplectic eigenvalues are strictly
greater than one (this also means that the state is positive definite). In our
paper, we focus exclusively on faithful Gaussian states.

We also define%
\begin{equation}
W_{\rho}=-V_{\rho}i\Omega, \label{eq:def-W}%
\end{equation}
which by the relations in \eqref{eq:cov-to-gibbs} and \eqref{eq:gibbs-to-cov}
gives us the following well known Cayley transforms
\cite{C46,C89,GV96}:
\begin{align}
\exp\left(  i\Omega H_{\rho}\right)   &  =\frac{W_{\rho}-I}{W_{\rho}%
+I},\label{eq:W-to-H}\\
W_{\rho}  &  =\frac{I+\exp\left(  i\Omega H_{\rho}\right)  }{I-\exp\left(
i\Omega H_{\rho}\right)  }. \label{eq:H-to-W}%
\end{align}
In the above and in what follows, our convention is that $\frac{A}{B}=AB^{-1}$ for
matrices $A$ and $B$, but observe that the ordering does not matter if $A$ and
$B$ commute. By substituting \eqref{eq:def-W} into
\eqref{eq:norm-of-Gaussian-form}, we see that%
\begin{equation}
\operatorname{Tr}\left\{  \exp\left[  -\frac{1}{2}\hat{x}^{T}H_{\rho}\hat
{x}\right]  \right\}  =\sqrt{\det(\left[  I-W\right]  i\Omega/2)}.
\label{eq:det-with-W}%
\end{equation}

The mean displacement $s_{\rho}\in\mathbb{R}^{2n}$ in
\eqref{eq: generic Gaussian} can be generated by applying the displacement
operator, defined for $s\in\mathbb{R}^{2n}$ as
\begin{equation}
D(s)=\exp\left[  s^{T}i\Omega\hat{x}\right]  =\exp\left[  -\hat{x}^{T}i\Omega
s\right]  , \label{eq: displacement Op}%
\end{equation}
on a zero-mean state $\exp\left[  -\frac{1}{2}\hat{x}^{T}H_{\rho}\hat
{x}\right]  /\sqrt{\det(\left[  V_{\rho}+i\Omega\right]  /2)}$\ as follows:%
\begin{equation}
D(-s_{\rho})\frac{\exp\left[  -\frac{1}{2}\hat{x}^{T}H_{\rho}\hat{x}\right]
}{\sqrt{\det(\left[  V_{\rho}+i\Omega\right]  /2)}}D(s_{\rho})=\frac
{\exp\left[  -\frac{1}{2}(\hat{x}-s_{\rho})^{T}H_{\rho}(\hat{x}-s_{\rho
})\right]  }{\sqrt{\det(\left[  V_{\rho}+i\Omega\right]  /2)}}.
\end{equation}
In our paper, we also consider the operator $D(s)$ in the more general case
when $s\in\mathbb{C}^{2n}$, but then we no longer refer to it as a
\textquotedblleft displacement operator\textquotedblright\ because it loses
its physical interpretation in this more general case.

\section{Computations with quantum Gaussian states}

\label{sec:computations-guassian}

\subsection{Powers of quantum Gaussian states}

\begin{proposition}
\label{prop:rho-to-alpha}Given a quantum Gaussian state $\rho$ expressed as%
\begin{equation}
\rho=\frac{1}{\left[  \det\left(  \left[  V+i\Omega\right]  /2\right)
\right]  ^{1/2}}\exp\left[  -\frac{1}{2}\hat{x}^{T}H\hat{x}\right]  ,
\end{equation}
for a positive-definite real matrix $H$ and corresponding covariance matrix
$V$, the density operator $\rho^{\alpha}/\operatorname{Tr}\{\rho^{\alpha}\}$,
for $\alpha>0$, can be written as%
\begin{equation}
\rho(\alpha)\equiv\frac{\rho^{\alpha}}{\operatorname{Tr}\{\rho^{\alpha}%
\}}=\frac{1}{\left[  \det\left(  \left[  V^{(\alpha)}+i\Omega\right]
/2\right)  \right]  ^{1/2}}\exp\left[  -\frac{1}{2}\hat{x}^{T}H^{(\alpha)}%
\hat{x}\right]  ,
\end{equation}
where the positive-definite real matrix $H^{(\alpha)}$ and the covariance
matrix $V^{(\alpha)}$ are given by%
\begin{align}
H^{(\alpha)}  &  =\alpha H,\\
V^{(\alpha)}  &  \equiv V_{\rho(\alpha)}\equiv\frac{\left(  I+\left(  V_{\rho
}i\Omega\right)  ^{-1}\right)  ^{\alpha}+\left(  I-\left(  V_{\rho}%
i\Omega\right)  ^{-1}\right)  ^{\alpha}}{\left(  I+\left(  V_{\rho}%
i\Omega\right)  ^{-1}\right)  ^{\alpha}-\left(  I-\left(  V_{\rho}%
i\Omega\right)  ^{-1}\right)  ^{\alpha}}i\Omega. \label{eq:alpha-cov-matrix}%
\end{align}

\end{proposition}

\begin{proof}
Consider that%
\begin{align}
\rho^{\alpha}  &  =\left[  \frac{1}{\left[  \det\left(  \left[  V+i\Omega
\right]  /2\right)  \right]  ^{1/2}}\exp\left[  -\frac{1}{2}\hat{x}^{T}%
H\hat{x}\right]  \right]  ^{\alpha}\\
&  =\frac{1}{\left[  \det\left(  \left[  V+i\Omega\right]  /2\right)  \right]
^{\alpha/2}}\exp\left[  -\frac{1}{2}\hat{x}^{T}\left[  \alpha H\right]
\hat{x}\right]  . \label{eq:rho-to-alpha}%
\end{align}
Then%
\begin{equation}
\operatorname{Tr}\{\rho^{\alpha}\}=\frac{1}{\left[  \det\left(  \left[
V+i\Omega\right]  /2\right)  \right]  ^{\alpha/2}}\operatorname{Tr}\left\{
\exp\left[  -\frac{1}{2}\hat{x}^{T}\left[  \alpha H\right]  \hat{x}\right]
\right\}  . \label{eq:rho-to-alpha-normalization}%
\end{equation}
To compute the covariance matrix corresponding to $\alpha H$, which we call
$V_{\rho(\alpha)}$, we exploit \eqref{eq:cov-to-gibbs} and
\eqref{eq:gibbs-to-cov} and find that%
\begin{align}
V_{\rho(\alpha)}  &  =\coth(i\Omega\alpha H_{\rho}/2)i\Omega\\
&  =\coth(i\Omega\alpha\left[  2i\Omega\operatorname{arcoth}(V_{\rho}%
i\Omega)\right]  /2)i\Omega\\
&  =\coth(\alpha\operatorname{arcoth}(V_{\rho}i\Omega))i\Omega.
\end{align}
To evaluate the last equality, consider for $\left\vert x\right\vert >1$ that%
\begin{align}
\coth(\alpha\operatorname{arcoth}(x))  &  =\coth\left(  \alpha\frac{1}{2}%
\ln\left(  \frac{x+1}{x-1}\right)  \right) \\
&  =\coth\left(  \alpha\frac{1}{2}\ln\left(  \frac{1+1/x}{1-1/x}\right)
\right) \\
&  =\coth\left(  \frac{1}{2}\ln\left[  \left(  \frac{1+1/x}{1-1/x}\right)
^{\alpha}\right]  \right) \\
&  =\frac{\exp\left(  2\left(  \frac{1}{2}\ln\left[  \left(  \frac
{1+1/x}{1-1/x}\right)  ^{\alpha}\right]  \right)  \right)  +1}{\exp\left(
2\left(  \frac{1}{2}\ln\left[  \left(  \frac{1+1/x}{1-1/x}\right)  ^{\alpha
}\right]  \right)  \right)  -1}\\
&  =\frac{\left(  \frac{1+1/x}{1-1/x}\right)  ^{\alpha}+1}{\left(
\frac{1+1/x}{1-1/x}\right)  ^{\alpha}-1}\\
&  =\frac{\left(  1+1/x\right)  ^{\alpha}+\left(  1-1/x\right)  ^{\alpha}%
}{\left(  1+1/x\right)  ^{\alpha}-\left(  1-1/x\right)  ^{\alpha}}.
\end{align}
Several of the above manipulations are possible because $1\pm1/x>0$ for
$\left\vert x\right\vert >1$, so that $y\rightarrow y^{\alpha}$ is a well
defined function from $\mathbb{R}^{+}\rightarrow\mathbb{R}^{+}$. The matrix
$V_{\rho}i\Omega$ has all of its eigenvalues $>1$ or $<-1$, so that the above
development applies, and we find that%
\begin{equation}
V_{\rho(\alpha)}=\frac{\left(  I+\left(  V_{\rho}i\Omega\right)  ^{-1}\right)
^{\alpha}+\left(  I-\left(  V_{\rho}i\Omega\right)  ^{-1}\right)  ^{\alpha}%
}{\left(  I+\left(  V_{\rho}i\Omega\right)  ^{-1}\right)  ^{\alpha}-\left(
I-\left(  V_{\rho}i\Omega\right)  ^{-1}\right)  ^{\alpha}}i\Omega.
\end{equation}
The matrix $V_{\rho(\alpha)}$ is a legitimate covariance matrix as a
consequence of \eqref{eq:H-PD-V-legit}\ and the fact that $\alpha H>0$. So we
conclude from \eqref{eq:norm-of-Gaussian-form}\ that%
\begin{equation}
\operatorname{Tr}\left\{  \exp\left[  -\frac{1}{2}\hat{x}^{T}\left[  \alpha
H\right]  \hat{x}\right]  \right\}  =\left[  \det\left(  \left[
V_{\rho(\alpha)}+i\Omega\right]  /2\right)  \right]  ^{1/2}.
\label{eq:rho-to-alpha-normalization-final}%
\end{equation}
Putting together \eqref{eq:rho-to-alpha},
\eqref{eq:rho-to-alpha-normalization}, and
\eqref{eq:rho-to-alpha-normalization-final} gives the statement of the proposition.
\end{proof}

\begin{remark}
The covariance matrix $V_{\rho(\alpha)}$ is equal to the $V(\alpha)$
covariance matrix given in Eqs.~(54)\ and (55)\ of Ref.~\onlinecite{PL08}. We give a
proof for this equality in Appendix~\ref{app:CM-for-rho-alpha}.
\end{remark}

Now we give alternative proofs of some results from Ref.~\onlinecite{PS00}, which can be
viewed as consequences of Proposition~\ref{prop:rho-to-alpha}:

\begin{corollary}
[{[\onlinecite{PS00}]}]\label{cor:rho-to-the-2}Given a quantum Gaussian state $\rho$
expressed as%
\begin{equation}
\rho=\frac{1}{\left[  \det\left(  \left[  V+i\Omega\right]  /2\right)
\right]  ^{1/2}}\exp\left[  -\frac{1}{2}\hat{x}^{T}H\hat{x}\right]  ,
\end{equation}
for a positive-definite real matrix $H$ and corresponding covariance matrix
$V$, the density operator $\rho^{2}/\operatorname{Tr}\{\rho^{2}\}$ can be
written as%
\begin{equation}
\frac{\rho^{2}}{\operatorname{Tr}\{\rho^{2}\}}=\frac{1}{\left[  \det\left(
\left[  V^{(2)}+i\Omega\right]  /2\right)  \right]  ^{1/2}}\exp\left[
-\frac{1}{2}\hat{x}^{T}H^{(2)}\hat{x}\right]  ,
\end{equation}
where the positive-definite real matrix $H^{(2)}$ and the covariance matrix
$V^{(2)}$ are given by%
\begin{align}
H^{(2)}  &  =2H,\label{eq:H2-2H}\\
V^{(2)}  &  =\frac{1}{2}\left(  V+\Omega V^{-1}\Omega^{T}\right)  .
\label{eq:W(2)-equality}%
\end{align}

\end{corollary}

\begin{proof}
Our starting point is the expression for $V^{(2)}$ in
\eqref{eq:alpha-cov-matrix}. Consider that the matrices in the numerator and
denominator are all commuting, so that we can work with the scalar function
$x\rightarrow\frac{\left(  1+1/x\right)  ^{\alpha}+\left(  1-1/x\right)
^{\alpha}}{\left(  1+1/x\right)  ^{\alpha}-\left(  1-1/x\right)  ^{\alpha}}$
for $\left\vert x\right\vert >1$ and simplify it for $\alpha=2$:%
\begin{align}
\frac{\left(  1+1/x\right)  ^{2}+\left(  1-1/x\right)  ^{2}}{\left(
1+1/x\right)  ^{2}-\left(  1-1/x\right)  ^{2}}  &  =\frac{1+2/x+1/x^{2}%
+1-2/x+1/x^{2}}{1+2/x+1/x^{2}-\left(  1-2/x+1/x^{2}\right)  }\\
&  =\frac{2+2/x^{2}}{4/x}\\
&  =\frac{1}{2}\left(  x+x^{-1}\right)  .
\end{align}
So we conclude that%
\begin{align}
V^{(2)}  &  =\frac{1}{2}\left[  Vi\Omega+\left(  Vi\Omega\right)
^{-1}\right]  i\Omega\\
&  =\frac{1}{2}\left[  V+i\Omega V^{-1}i\Omega\right] \\
&  =\frac{1}{2}\left[  V+\Omega V^{-1}\Omega^{T}\right]  .
\end{align}
Proposition~\ref{prop:rho-to-alpha} already justified that the matrix
$V^{(2)}$ is a legitimate covariance matrix. From
Proposition~\ref{prop:rho-to-alpha}, we know that $H^{(2)}=2H$.
\end{proof}

\begin{corollary}
[{[\onlinecite{PS00}]}]\label{prop:rho-to-1/2}Given a quantum Gaussian state $\rho$
expressed as%
\begin{equation}
\rho=\frac{1}{\left[  \det\left(  \left[  V+i\Omega\right]  /2\right)
\right]  ^{1/2}}\exp\left[  -\frac{1}{2}\hat{x}^{T}H\hat{x}\right]  ,
\end{equation}
for a positive-definite real matrix $H$ and corresponding covariance matrix
$V$, the density operator $\rho^{1/2}/\operatorname{Tr}\{\rho^{1/2}\}$ can be
written as%
\begin{equation}
\rho(1/2)=\frac{\rho^{1/2}}{\operatorname{Tr}\{\rho^{1/2}\}}=\frac{1}{\left[
\det\left(  \left[  V^{(1/2)}+i\Omega\right]  /2\right)  \right]  ^{1/2}}%
\exp\left[  -\frac{1}{2}\hat{x}^{T}H^{(1/2)}\hat{x}\right]  ,
\end{equation}
where the positive-definite real matrix $H^{(1/2)}$ and the covariance matrix
$V^{(1/2)}$ are given by%
\begin{align}
H^{(1/2)}  &  =H/2,\\
V^{(1/2)}  &  =V_{\rho(1/2)}=\left(  \sqrt{I+(V\Omega)^{-2}}+I\right)  V.
\end{align}

\end{corollary}

\begin{proof}
Our starting point is the expression for $V^{(1/2)}$ in
\eqref{eq:alpha-cov-matrix}. Consider that the matrices in the numerator and
denominator are all commuting, so that we can work with the scalar function
$x\rightarrow\frac{\left(  1+1/x\right)  ^{\alpha}+\left(  1-1/x\right)
^{\alpha}}{\left(  1+1/x\right)  ^{\alpha}-\left(  1-1/x\right)  ^{\alpha}}$
for $\left\vert x\right\vert >1$ and simplify it for $\alpha=1/2$:%
\begin{align}
&  \frac{\left(  1+1/x\right)  ^{1/2}+\left(  1-1/x\right)  ^{1/2}}{\left(
1+1/x\right)  ^{1/2}-\left(  1-1/x\right)  ^{1/2}}\nonumber\\
&  =\left[  \frac{\left(  1+1/x\right)  ^{1/2}+\left(  1-1/x\right)  ^{1/2}%
}{\left(  1+1/x\right)  ^{1/2}-\left(  1-1/x\right)  ^{1/2}}\right]  \left[
\frac{\left(  1+1/x\right)  ^{1/2}+\left(  1-1/x\right)  ^{1/2}}{\left(
1+1/x\right)  ^{1/2}+\left(  1-1/x\right)  ^{1/2}}\right] \\
&  =\frac{1+1/x+2\sqrt{\left(  1+1/x\right)  \left(  1-1/x\right)  }%
+1-1/x}{1+1/x-\left(  1-1/x\right)  }\\
&  =\frac{2+2\sqrt{1-1/x^{2}}}{2/x}\\
&  =\left(  1+\sqrt{1-1/x^{2}}\right)  x.
\end{align}
So we conclude that%
\begin{equation}
V^{(1/2)}=\left(  I+\sqrt{I-\left(  Vi\Omega\right)  ^{-2}}\right)  \left(
Vi\Omega\right)  i\Omega=\left(  \sqrt{I+(V\Omega)^{-2}}+I\right)  V.
\end{equation}
Proposition~\ref{prop:rho-to-alpha} already justified that the matrix
$V^{(1/2)}$ is a legitimate covariance matrix. From
Proposition~\ref{prop:rho-to-alpha}, we know that $H^{(1/2)}=H/2$.
\end{proof}

\subsection{Traces of compositions of quantum Gaussian states}

In what follows, we repeatedly make use of the following well known lemma:

\begin{lemma}
[{[\onlinecite{W50}]}]\label{lem:woodbury}Given invertible matrices $A$ and $B$, the
following equality holds%
\begin{equation}
\left(  A+B\right)  ^{-1}=A^{-1}-A^{-1}\left(  A^{-1}+B^{-1}\right)
^{-1}A^{-1}.
\end{equation}

\end{lemma}

\begin{proposition}
[{[\onlinecite{Balian1969},\ \onlinecite{PS00}]}]\label{prop:product-rule-Gaussian}Given symmetric
matrices $H_{1}$ and $H_{2}$, the following equality holds%
\begin{equation}
\exp\left[  -\frac{1}{2}\hat{x}^{T}H_{1}\hat{x}\right]  \exp\left[  -\frac
{1}{2}\hat{x}^{T}H_{2}\hat{x}\right]  =\exp\left[  -\frac{1}{2}\hat{x}%
^{T}H_{3}\hat{x}\right]  , \label{eq:product-Gaussian-forms}%
\end{equation}
where $H_{3}$ is a symmetric matrix such that%
\begin{align}
H_{3}  &  =2i\Omega\operatorname{arcoth}(V_{3}i\Omega),\\
V_{3}  &  =-i\Omega+\left(  V_{2}+i\Omega\right)  \left(  V_{2}+V_{1}\right)
^{-1}\left(  V_{1}+i\Omega\right)  ,\label{eq:V3-product}\\
V_{1}  &  =\coth(i\Omega H_{1}/2)i\Omega,\\
V_{2}  &  =\coth(i\Omega H_{2}/2)i\Omega.
\end{align}

\end{proposition}

\begin{proof}
The equality in \eqref{eq:product-Gaussian-forms} is one of the main results
in Ref.~\onlinecite{Balian1969}, and the particular form of $V_{3}$ in
\eqref{eq:V3-product} was determined in Ref.~\onlinecite{PS00}. From Ref.~\onlinecite{Balian1969},
we know that $H_{3}$ is a symmetric matrix, and furthermore, the matrix
$H_{3}$ that satisfies \eqref{eq:product-Gaussian-forms} is the same one that
satisfies the following equation:%
\begin{equation}
\exp\left[  -i\Omega H_{1}\right]  \exp\left[  -i\Omega H_{2}\right]
=\exp\left[  -i\Omega H_{3}\right]  .
\end{equation}
Note that, by taking inverses, this latter equation is equivalent to%
\begin{equation}
\exp\left[  i\Omega H_{3}\right]  =\exp\left[  i\Omega H_{2}\right]
\exp\left[  i\Omega H_{1}\right]  . \label{eq:exp-i-Om-compose}%
\end{equation}
We use the relations in \eqref{eq:cov-to-gibbs}, \eqref{eq:gibbs-to-cov},
and\ \eqref{eq:H-to-W}\ to relate $H_{3}$ to matrices $V_{3}$ and $W_{3}$
given by%
\begin{equation}
V_{3}=-W_{3}i\Omega,
\end{equation}
where%
\begin{equation}
W_{3}=\frac{I+\exp\left(  i\Omega H_{3}\right)  }{I-\exp\left(  i\Omega
H_{3}\right)  }. \label{eq:W3-start}%
\end{equation}
For convenience of the reader, we detail some algebraic manipulations that
lead to the form of $V_{3}$ in \eqref{eq:V3-product}, but we note that it is
possible to arrive at this form by other means \cite{PS00}. By
\eqref{eq:W-to-H}, we have that
\begin{align}
W_{3}  &  =\frac{I+\exp\left(  i\Omega H_{3}\right)  }{I-\exp\left(  i\Omega
H_{3}\right)  }\\
&  =\left(  I+\exp\left[  i\Omega H_{2}\right]  \exp\left[  i\Omega
H_{1}\right]  \right)  \left(  I-\exp\left[  i\Omega H_{2}\right]  \exp\left[
i\Omega H_{1}\right]  \right)  ^{-1}\\
&  =\left(  \exp\left[  i\Omega H_{2}\right]  +\exp\left[  -i\Omega
H_{1}\right]  \right)  \left(  \exp\left[  -i\Omega H_{1}\right]  -\exp\left[
i\Omega H_{2}\right]  \right)  ^{-1}\label{eq:use-later-sqrt-sand}\\
&  =\left(  \exp\left[  i\Omega H_{2}\right]  -\exp\left[  -i\Omega
H_{1}\right]  +2\exp\left[  -i\Omega H_{1}\right]  \right)  \left(
\exp\left[  -i\Omega H_{1}\right]  -\exp\left[  i\Omega H_{2}\right]  \right)
^{-1}\\
&  =-I-2\exp\left[  -i\Omega H_{1}\right]  \left(  \exp\left[  i\Omega
H_{2}\right]  -\exp\left[  -i\Omega H_{1}\right]  \right)  ^{-1}\\
&  =-I-2\exp\left[  -i\Omega H_{1}\right]  \left(  \exp\left[  i\Omega
H_{2}\right]  -I-\left(  \exp\left[  -i\Omega H_{1}\right]  -I\right)
\right)  ^{-1}.
\end{align}
Consider from \eqref{eq:W-to-H} that%
\begin{align}
\exp\left[  i\Omega H_{2}\right]  -I  &  =\frac{W_{2}-I}{W_{2}+I}-\left[
\frac{W_{2}+I}{W_{2}+I}\right]  =-2\left[  W_{2}+I\right]  ^{-1},\\
\exp\left[  -i\Omega H_{1}\right]  -I  &  =\frac{W_{1}+I}{W_{1}-I}-\left[
\frac{W_{1}-I}{W_{1}-I}\right]  =2\left[  W_{1}-I\right]  ^{-1}.
\end{align}
So we find that%
\begin{align}
W_{3}  &  =-I-2\exp\left[  -i\Omega H_{1}\right]  \left(  -2\left[
W_{2}+I\right]  ^{-1}-2\left[  W_{1}-I\right]  ^{-1}\right)  ^{-1}\\
&  =-I+\exp\left[  -i\Omega H_{1}\right]  \left(  \left[  W_{2}+I\right]
^{-1}+\left[  W_{1}-I\right]  ^{-1}\right)  ^{-1}\\
&  =-I+\frac{W_{1}+I}{W_{1}-I}\left(  \left[  W_{2}+I\right]  ^{-1}+\left[
W_{1}-I\right]  ^{-1}\right)  ^{-1}.
\end{align}
Applying Lemma~\ref{lem:woodbury}\ with $A=\left(  W_{1}-I\right)  ^{-1}$ and
$B=\left(  W_{2}+I\right)  ^{-1}$, we find that%
\begin{align}
\left(  \left[  W_{2}+I\right]  ^{-1}+\left[  W_{1}-I\right]  ^{-1}\right)
^{-1}  &  =\left[  W_{1}-I\right]  -\left[  W_{1}-I\right]  \left(
W_{1}-I+W_{2}+I\right)  ^{-1}\left[  W_{1}-I\right] \\
&  =\left[  W_{1}-I\right]  -\left[  W_{1}-I\right]  \left(  W_{1}%
+W_{2}\right)  ^{-1}\left[  W_{1}-I\right]  ,
\end{align}
and this implies that%
\begin{align}
W_{3}  &  =-I+\frac{W_{1}+I}{W_{1}-I}\left(  \left[  W_{1}-I\right]  -\left[
W_{1}-I\right]  \left(  W_{1}+W_{2}\right)  ^{-1}\left[  W_{1}-I\right]
\right) \\
&  =-I+W_{1}+I-\left[  W_{1}+I\right]  \left(  W_{1}+W_{2}\right)
^{-1}\left[  W_{1}-I\right] \\
&  =W_{1}-\left[  W_{1}+I\right]  \left(  W_{1}+W_{2}\right)  ^{-1}\left[
W_{1}-I\right]  . \label{eq:use-later-sqrt-sand-last}%
\end{align}
Continuing, we have that%
\begin{align}
W_{3}  &  =W_{1}-\left[  W_{1}+W_{2}-W_{2}+I\right]  \left(  W_{1}%
+W_{2}\right)  ^{-1}\left[  W_{1}-I\right] \\
&  =W_{1}-\left[  W_{1}+W_{2}\right]  \left(  W_{1}+W_{2}\right)  ^{-1}\left[
W_{1}-I\right]  -\left[  -W_{2}+I\right]  \left(  W_{1}+W_{2}\right)
^{-1}\left[  W_{1}-I\right] \\
&  =W_{1}-\left[  W_{1}-I\right]  +\left[  W_{2}-I\right]  \left(  W_{1}%
+W_{2}\right)  ^{-1}\left[  W_{1}-I\right] \\
&  =I+\left[  W_{2}-I\right]  \left(  W_{1}+W_{2}\right)  ^{-1}\left[
W_{1}-I\right]  . \label{eq:W-composition-rule}%
\end{align}
So this implies, from \eqref{eq:def-W}, that%
\begin{align}
-V_{3}i\Omega &  =I+\left(  -V_{2}i\Omega-I\right)  \left(  -V_{2}%
i\Omega-V_{1}i\Omega\right)  ^{-1}\left(  -V_{1}i\Omega-I\right) \\
&  =I-\left(  V_{2}i\Omega+I\right)  \left(  \left[  V_{2}+V_{1}\right]
i\Omega\right)  ^{-1}\left(  V_{1}i\Omega+I\right) \\
&  =I-\left(  V_{2}i\Omega+I\right)  i\Omega\left(  V_{2}+V_{1}\right)
^{-1}\left(  V_{1}i\Omega+I\right) \\
&  =I-\left(  V_{2}+i\Omega\right)  \left(  V_{2}+V_{1}\right)  ^{-1}\left(
V_{1}i\Omega+I\right)  .
\end{align}
This finally implies \eqref{eq:V3-product}.
\end{proof}

\begin{lemma}
The matrix $V_{3}$ from Proposition~\ref{prop:product-rule-Gaussian} is
symmetric, which follows from the fact that $H_{3}$ is symmetric or by
inspecting the following identity:%
\begin{multline}
-i\Omega+\left(  V_{2}+i\Omega\right)  \left(  V_{2}+V_{1}\right)
^{-1}\left(  V_{1}+i\Omega\right) \\
=\left(  V_{2}^{-1}+V_{1}^{-1}\right)  ^{-1}-V_{1}\left(  V_{2}+V_{1}\right)
^{-1}i\Omega+i\Omega\left(  V_{2}+V_{1}\right)  ^{-1}V_{1}+\Omega\left(
V_{2}+V_{1}\right)  ^{-1}\Omega^{T}.
\end{multline}

\end{lemma}

\begin{proof}
Consider that%
\begin{align}
&  -i\Omega+\left(  V_{2}+i\Omega\right)  \left(  V_{2}+V_{1}\right)
^{-1}\left(  V_{1}+i\Omega\right) \nonumber\\
&  =-i\Omega+\left(  V_{2}+V_{1}-V_{1}+i\Omega\right)  \left(  V_{2}%
+V_{1}\right)  ^{-1}\left(  V_{1}+i\Omega\right) \\
&  =-i\Omega+\left(  V_{2}+V_{1}\right)  \left(  V_{2}+V_{1}\right)
^{-1}\left(  V_{1}+i\Omega\right)  +\left(  -V_{1}+i\Omega\right)  \left(
V_{2}+V_{1}\right)  ^{-1}\left(  V_{1}+i\Omega\right) \\
&  =-i\Omega+\left(  V_{1}+i\Omega\right)  +\left(  -V_{1}+i\Omega\right)
\left(  V_{2}+V_{1}\right)  ^{-1}\left(  V_{1}+i\Omega\right) \\
&  =V_{1}+\left(  -V_{1}+i\Omega\right)  \left(  V_{2}+V_{1}\right)
^{-1}\left(  V_{1}+i\Omega\right) \\
&  =V_{1}-V_{1}\left(  V_{2}+V_{1}\right)  ^{-1}V_{1}-V_{1}\left(  V_{2}%
+V_{1}\right)  ^{-1}i\Omega+i\Omega\left(  V_{2}+V_{1}\right)  ^{-1}%
V_{1}\nonumber\\
&  \qquad+i\Omega\left(  V_{2}+V_{1}\right)  ^{-1}i\Omega\\
&  =\left(  V_{2}^{-1}+V_{1}^{-1}\right)  ^{-1}-V_{1}\left(  V_{2}%
+V_{1}\right)  ^{-1}i\Omega+i\Omega\left(  V_{2}+V_{1}\right)  ^{-1}%
V_{1}+\Omega\left(  V_{2}+V_{1}\right)  ^{-1}\Omega^{T}.
\end{align}
In the last line, we used Lemma~\ref{lem:woodbury} with $A=V_{1}^{-1}$ and
$B=V_{2}^{-1}$ and the fact that $\Omega^{T}=-\Omega$.
\end{proof}

\begin{proposition}
[{[\onlinecite{Balian1969},\ \onlinecite{BBP15},\ \onlinecite{LDW17}]}]\label{prop:sandwich-with-square-root}Given
positive-definite real matrices $H_{4}$ and $H_{5}$, we have that%
\begin{equation}
\exp\left[  -\frac{1}{2}\hat{x}^{T}\left[  H_{4}/2\right]  \hat{x}\right]
\exp\left[  -\frac{1}{2}\hat{x}^{T}H_{5}\hat{x}\right]  \exp\left[  -\frac
{1}{2}\hat{x}^{T}\left[  H_{4}/2\right]  \hat{x}\right]  =\exp\left[
-\frac{1}{2}\hat{x}^{T}H_{6}\hat{x}\right]  ,
\end{equation}
where $H_{6}$ is a positive-definite real matrix with corresponding covariance matrix $V_6$, given by%
\begin{align}
H_{6}  &  =2i\Omega\operatorname{arcoth}(V_{6}i\Omega),\\
V_{6}  &  =V_{4}-\left(  \sqrt{I+\left(  V_{4}\Omega\right)  ^{-2}}\right)
V_{4}\left(  V_{5}+V_{4}\right)  ^{-1}V_{4}\left(  \sqrt{I+\left(  \Omega
V_{4}\right)  ^{-2}}\right)  .
\end{align}

\end{proposition}

\begin{proof}
Let $H_{7}=H_{4}/2$ and let $V_{7}$ be the covariance matrix defined by%
\begin{equation}
V_{7}=\coth(i\Omega H_{7}/2)i\Omega.
\end{equation}
From Corollary~\ref{prop:rho-to-1/2}, it follows that $V_{7}$ can be given in
terms of $V_{4}$ as%
\begin{equation}
V_{7}=\left(  \sqrt{I+\left(  V_{4}\Omega\right)  ^{-2}}+I\right)  V_{4},
\end{equation}
which is equivalent to%
\begin{equation}
W_{7}=\left(  \sqrt{I-W_{4}^{-2}}+I\right)  W_{4}.
\end{equation}

Consider from two applications of the composition rule in
\eqref{eq:exp-i-Om-compose}\ that%
\begin{equation}
\exp(i\Omega H_{6})=\exp(i\Omega H_{4}/2)\exp(i\Omega H_{5})\exp(i\Omega
H_{4}/2).
\end{equation}
This implies from \eqref{eq:W-to-H} that%
\begin{align}
W_{6}  &  =\left[  I+e^{i\Omega H_{4}/2}e^{i\Omega H_{5}}e^{i\Omega H_{4}%
/2}\right]  \left[  I-e^{i\Omega H_{4}/2}e^{i\Omega H_{5}}e^{i\Omega H_{4}%
/2}\right]  ^{-1}\label{eq:sandwich-W-start}\\
&  =\left[  e^{-i\Omega H_{4}/2}+e^{i\Omega H_{4}/2}e^{i\Omega H_{5}}\right]
\left[  e^{-i\Omega H_{4}/2}-e^{i\Omega H_{4}/2}e^{i\Omega H_{5}}\right]
^{-1}\\
&  =\exp(i\Omega H_{4}/2)\left[  e^{-i\Omega H_{4}}+e^{i\Omega H_{5}}\right]
\left(  \exp(i\Omega H_{4}/2)\left[  e^{-i\Omega H_{4}}-e^{i\Omega H_{5}%
}\right]  \right)  ^{-1}\\
&  =\exp(i\Omega H_{4}/2)\left[  e^{-i\Omega H_{4}}+e^{i\Omega H_{5}}\right]
\left[  e^{-i\Omega H_{4}}-e^{i\Omega H_{5}}\right]  ^{-1}\exp(-i\Omega
H_{4}/2).
\end{align}
From the development in
\eqref{eq:use-later-sqrt-sand}--\eqref{eq:use-later-sqrt-sand-last}, we know
that%
\begin{multline}
\left[  \exp(-i\Omega H_{4})+\exp(i\Omega H_{5})\right]  \left(  \exp(-i\Omega
H_{4})-\exp(i\Omega H_{5})\right)  ^{-1}\\
=W_{4}-\left[  W_{4}+I\right]  \left[  W_{4}+W_{5}\right]  ^{-1}\left[
W_{4}-I\right]  .
\end{multline}
This implies that%
\begin{align}
W_{6}  &  =\exp(i\Omega H_{4}/2)\left[  W_{4}-\left[  W_{4}+I\right]  \left[
W_{4}+W_{5}\right]  ^{-1}\left[  W_{4}-I\right]  \right]  \exp(-i\Omega
H_{4}/2)\\
&  =\exp(i\Omega H_{4}/2)W_{4}\exp(-i\Omega H_{4}/2)-\exp(i\Omega
H_{4}/2)\left[  W_{4}+I\right]  \left[  W_{4}+W_{5}\right]  ^{-1}\left[
W_{4}-I\right]  \exp(-i\Omega H_{4}/2)\\
&  =W_{4}-\exp(i\Omega H_{4}/2)\left[  W_{4}+I\right]  \left[  W_{4}%
+W_{5}\right]  ^{-1}\left[  W_{4}-I\right]  \exp(-i\Omega H_{4}/2)\\
&  =W_{4}-\frac{\left(  \sqrt{I-W_{4}^{-2}}+I\right)  W_{4}-I}{\left(
\sqrt{I-W_{4}^{-2}}+I\right)  W_{4}+I}\left[  W_{4}+I\right]  \left[
W_{4}+W_{5}\right]  ^{-1}\left[  W_{4}-I\right]  \frac{\left(  \sqrt
{I-W_{4}^{-2}}+I\right)  W_{4}+I}{\left(  \sqrt{I-W_{4}^{-2}}+I\right)
W_{4}-I},
\end{align}
where the last equality follows from \eqref{eq:W-to-H} and
Corollary~\ref{prop:rho-to-1/2}. Considering that the following scalar
functions simplify as%
\begin{align}
\frac{\left(  \sqrt{1-x^{-2}}+1\right)  x-1}{\left(  \sqrt{1-x^{-2}}+1\right)
x+1}\left[  x+1\right]   &  =\left(  \sqrt{1-x^{-2}}\right)  x,\\
\left[  x-1\right]  \frac{\left(  \sqrt{1-x^{-2}}+1\right)  x+1}{\left(
\sqrt{1-x^{-2}}+1\right)  x-1}  &  =x\sqrt{1-x^{-2}},
\end{align}
we find that%
\begin{equation}
W_{6}=W_{4}-\left(  \sqrt{I-W_{4}^{-2}}\right)  W_{4}\left[  W_{4}%
+W_{5}\right]  ^{-1}W_{4}\left(  \sqrt{I-W_{4}^{-2}}\right)  .
\label{eq:W-sandwich-sqrt}%
\end{equation}

Now substituting, we find that%
\begin{align}
&  -V_{6}i\Omega\nonumber\\
&  =-V_{4}i\Omega-\left(  \sqrt{I-\left(  -V_{4}i\Omega\right)  ^{-2}}\right)
\left(  -V_{4}i\Omega\right)  \left(  -V_{5}i\Omega-V_{4}i\Omega\right)
^{-1}\left(  -V_{4}i\Omega\right)  \left(  \sqrt{I-\left(  -V_{4}%
i\Omega\right)  ^{-2}}\right) \\
&  =-V_{4}i\Omega+\left(  \sqrt{I+\left(  V_{4}\Omega\right)  ^{-2}}\right)
\left(  V_{4}i\Omega\right)  \left(  \left[  V_{5}+V_{4}\right]
i\Omega\right)  ^{-1}\left(  V_{4}i\Omega\right)  \left(  \sqrt{I+\left(
V_{4}\Omega\right)  ^{-2}}\right) \\
&  =-V_{4}i\Omega+\left(  \sqrt{I+\left(  V_{4}\Omega\right)  ^{-2}}\right)
V_{4}\left(  V_{5}+V_{4}\right)  ^{-1}\left(  V_{4}i\Omega\right)  \left(
\sqrt{I+\left(  V_{4}\Omega\right)  ^{-2}}\right)  .
\end{align}
This then implies that%
\begin{align}
V_{6}  &  =V_{4}-\left(  \sqrt{I+\left(  V_{4}\Omega\right)  ^{-2}}\right)
V_{4}\left(  V_{5}+V_{4}\right)  ^{-1}\left(  V_{4}i\Omega\right)  \left(
\sqrt{I+\left(  V_{4}\Omega\right)  ^{-2}}\right)  i\Omega\\
&  =V_{4}-\left(  \sqrt{I+\left(  V_{4}\Omega\right)  ^{-2}}\right)
V_{4}\left(  V_{5}+V_{4}\right)  ^{-1}V_{4}\left(  \sqrt{I+\left(  \Omega
V_{4}\right)  ^{-2}}\right)  ,
\end{align}
which concludes the proof.
\end{proof}

Even though it directly follows from the above that $V_6$ is a legitimate covariance matrix, the following proposition gives an alternative confirmation of this fact:

\begin{proposition}
\label{prop:norm-sandwich-sqrt}The matrix $V_{6}$ from
Proposition~\ref{prop:sandwich-with-square-root}\ is a legitimate covariance
matrix, and so we can conclude that%
\begin{equation}
\operatorname{Tr}\left\{  \exp\left[  -\frac{1}{2}\hat{x}^{T}H_{6}\hat
{x}\right]  \right\}  =\sqrt{\det(V_{6}+i\Omega/2)},
\end{equation}
where $H_{6}=2i\Omega\operatorname{arcoth}(V_{\rho}i\Omega)$.
\end{proposition}

\begin{proof}
Since $V_{5}$ is a legitimate covariance matrix corresponding to
positive-definite real $H_{5}$, we have that $V_{5}-i\Omega>0$ which implies
that $V_{5}+V_{4}>V_{4}+i\Omega$ and in turn that $-\left(  V_{5}%
+V_{4}\right)  ^{-1}>-\left(  V_{4}+i\Omega\right)  ^{-1}$, by operator
monotonicity of the function $x\rightarrow-x^{-1}$. Then we find that%
\begin{align}
V_{6}-i\Omega &  =V_{4}-\left(  \sqrt{I+\left(  V_{4}\Omega\right)  ^{-2}%
}\right)  V_{4}\left(  V_{5}+V_{4}\right)  ^{-1}V_{4}\left(  \sqrt{I+\left(
\Omega V_{4}\right)  ^{-2}}\right)  -i\Omega\\
&  >V_{4}-\left(  \sqrt{I+\left(  V_{4}\Omega\right)  ^{-2}}\right)
V_{4}\left(  V_{4}+i\Omega\right)  ^{-1}V_{4}\left(  \sqrt{I+\left(  \Omega
V_{4}\right)  ^{-2}}\right)  -i\Omega\\
&  =V_{4}-\left(  V_{4}-i\Omega\right)  -i\Omega\\
&  =0.
\end{align}
In the above, the second equality follows from
Lemma~\ref{lem:simplify-sqrt-sandwich} below.
\end{proof}

\begin{lemma}
[{[\onlinecite{LDW17}]}]\label{lem:simplify-sqrt-sandwich}The following identity holds
for a covariance matrix $V$ such that $V+i\Omega>0$:%
\begin{equation}
\sqrt{I+(V\Omega)^{-2}}V(V+i\Omega)^{-1}V\sqrt{I+(\Omega V)^{-2}}=V-i\Omega.
\end{equation}

\end{lemma}

\begin{proof}
Consider that%
\begin{align}
&  \sqrt{I+(V\Omega)^{-2}}V(V+i\Omega)^{-1}V\sqrt{I+(\Omega V)^{-2}%
}\nonumber\\
&  =\sqrt{I+(V\Omega)^{-2}}Vi\Omega(Vi\Omega+I)^{-1}V\sqrt{I+(\Omega V)^{-2}%
}\\
&  =\sqrt{I-(Vi\Omega)^{-2}}Vi\Omega(Vi\Omega+I)^{-1}Vi\Omega i\Omega
\sqrt{I-(i\Omega V)^{-2}}i\Omega i\Omega\\
&  =\sqrt{I-(Vi\Omega)^{-2}}Vi\Omega(Vi\Omega+I)^{-1}Vi\Omega\sqrt{I-(i\Omega
i\Omega Vi\Omega Vi\Omega)^{-1}}i\Omega\\
&  =\left[  \sqrt{I-(Vi\Omega)^{-2}}Vi\Omega(Vi\Omega+I)^{-1}Vi\Omega
\sqrt{I-(Vi\Omega)^{-2}}\right]  i\Omega.
\end{align}
Now that the expression in square brackets has been reduced to a matrix
version of the scalar function $x\rightarrow\sqrt{1-x^{-2}}x(x+1)^{-1}%
x\sqrt{1-x^{-2}}$, we can use the fact that the scalar function collapses as%
\begin{equation}
\sqrt{1-x^{-2}}x(x+1)^{-1}x\sqrt{1-x^{-2}}=x-1,
\end{equation}
and we find that%
\begin{align}
&  \sqrt{I-(Vi\Omega)^{-2}}Vi\Omega(Vi\Omega+I)^{-1}Vi\Omega\sqrt
{I-(Vi\Omega)^{-2}}i\Omega\nonumber\\
&  =\left[  Vi\Omega-I\right]  i\Omega\\
&  =V-i\Omega.
\end{align}
This concludes the proof.
\end{proof}

The following proposition is again a consequence of Ref.~\onlinecite{Balian1969}, and the
particular form of the determinant in \eqref{eq:det-from-product-H1-H2}\ was
reported in Ref.~\onlinecite[Eq.~(3.5)]{MM12}.

\begin{proposition}
[{[\onlinecite{Balian1969},\ \onlinecite{MM12}]}]\label{prop:simple-product-normalization}Given
positive-definite real matrices $H_{1}$ and $H_{2}$, it follows that%
\begin{equation}
\exp\left[  -\frac{1}{2}\hat{x}^{T}H_{1}\hat{x}\right]  \exp\left[  -\frac
{1}{2}\hat{x}^{T}H_{2}\hat{x}\right]  =\exp\left[  -\frac{1}{2}\hat{x}%
^{T}H_{3}\hat{x}\right]  ,
\end{equation}
where $H_{3}$ is such that%
\begin{align}
\operatorname{Tr}\left\{  \exp\left[  -\frac{1}{2}\hat{x}^{T}H_{3}\hat
{x}\right]  \right\}   &  =\sqrt{\frac{\det(\left[  V_{1}+i\Omega\right]
/2)\det(\left[  V_{2}+i\Omega\right]  /2)}{\det(\left[  V_{1}+V_{2}\right]
/2)}},\label{eq:det-from-product-H1-H2}\\
V_{1}  &  =\coth(i\Omega H_{1}/2)i\Omega,\\
V_{2}  &  =\coth(i\Omega H_{2}/2)i\Omega.
\end{align}

\end{proposition}

\begin{proof}
Consider that%
\begin{align}
&  \operatorname{Tr}\left\{  \exp\left[  -\frac{1}{2}\hat{x}^{T}H_{1}\hat
{x}\right]  \exp\left[  -\frac{1}{2}\hat{x}^{T}H_{2}\hat{x}\right]  \right\}
\nonumber\\
&  =\operatorname{Tr}\left\{  \exp\left[  -\frac{1}{2}\hat{x}^{T}\left[
H_{1}/2\right]  \hat{x}\right]  \exp\left[  -\frac{1}{2}\hat{x}^{T}H_{2}%
\hat{x}\right]  \exp\left[  -\frac{1}{2}\hat{x}^{T}\left[  H_{1}/2\right]
\hat{x}\right]  \right\} \\
&  =\left[  \det\left(  \left[  V_{1}-\left(  \sqrt{I+\left(  V_{1}%
\Omega\right)  ^{-2}}\right)  V_{1}\left(  V_{2}+V_{1}\right)  ^{-1}%
V_{1}\left(  \sqrt{I+\left(  \Omega V_{1}\right)  ^{-2}}\right)
+i\Omega\right]  /2\right)  \right]  ^{1/2}\\
&  =\left[  \det\left(  \left[  I-W_{1}+\left(  \sqrt{I-W_{1}^{-2}}\right)
W_{1}\left(  W_{2}+W_{1}\right)  ^{-1}W_{1}\left(  \sqrt{I-W_{1}^{-2}}\right)
\right]  i\Omega/2\right)  \right]  ^{1/2}\\
&  =\left[  \det\left(  I-W_{1}+\left(  \sqrt{I-W_{1}^{-2}}\right)
W_{1}\left(  W_{2}+W_{1}\right)  ^{-1}W_{1}\left(  \sqrt{I-W_{1}^{-2}}\right)
\right)  \det\left(  i\Omega/2\right)  \right]  ^{1/2}.
\end{align}
The first equality follows from cyclicity of trace. The second equality
follows from Proposition~\ref{prop:norm-sandwich-sqrt}. The third equality
follows from \eqref{eq:W-sandwich-sqrt}\ and \eqref{eq:det-with-W}. We now
prove that the following matrices are similar%
\begin{align}
W^{\prime}  &  =W_{1}-\left(  \sqrt{I-W_{1}^{-2}}\right)  W_{1}\left(
W_{2}+W_{1}\right)  ^{-1}W_{1}\left(  \sqrt{I-W_{1}^{-2}}\right)  ,\\
W^{\prime\prime}  &  =I+\left(  W_{2}-I\right)  \left(  W_{2}+W_{1}\right)
^{-1}\left(  W_{1}-I\right)  ,
\end{align}
i.e., related as%
\begin{equation}
W^{\prime}=\exp(i\Omega H_{1}/2)W^{\prime\prime}\exp(-i\Omega H_{1}/2).
\end{equation}
To this end, consider from
\eqref{eq:exp-i-Om-compose}--\eqref{eq:W-composition-rule} that%
\begin{equation}
W^{\prime\prime}=\frac{I+\exp(i\Omega H_{2})\exp(i\Omega H_{1})}%
{I-\exp(i\Omega H_{2})\exp(i\Omega H_{1})},
\end{equation}
and from applying \eqref{eq:exp-i-Om-compose} twice and considering
\eqref{eq:sandwich-W-start}--\eqref{eq:W-sandwich-sqrt},%
\begin{equation}
W^{\prime}=\frac{I+\exp(i\Omega H_{1}/2)\exp(i\Omega H_{2})\exp(i\Omega
H_{1}/2)}{I-\exp(i\Omega H_{1}/2)\exp(i\Omega H_{2})\exp(i\Omega H_{1}/2)}.
\end{equation}
Then we find that%
\begin{align}
\exp(i\Omega H_{1}/2)W^{\prime\prime}\exp(-i\Omega H_{1}/2)  &  =\exp(i\Omega
H_{1}/2)\frac{I+\exp(i\Omega H_{2})\exp(i\Omega H_{1})}{I-\exp(i\Omega
H_{2})\exp(i\Omega H_{1})}\exp(-i\Omega H_{1}/2)\\
&  =\frac{I+\exp(i\Omega H_{1}/2)\exp(i\Omega H_{2})\exp(i\Omega H_{1}%
/2)}{I-\exp(i\Omega H_{1}/2)\exp(i\Omega H_{2})\exp(i\Omega H_{1}/2)}\\
&  =W^{\prime}.
\end{align}
Since these matrices are related by a similarity transformation, we find that%
\begin{align}
&  \left[  \det\left(  I-W_{1}+\left(  \sqrt{I-W_{1}^{-2}}\right)
W_{1}\left(  W_{2}+W_{1}\right)  ^{-1}W_{1}\left(  \sqrt{I-W_{1}^{-2}}\right)
\right)  \det\left(  i\Omega/2\right)  \right]  ^{1/2}\nonumber\\
&  =\left[  \det\left(  I-\left[  I+\left(  W_{2}-I\right)  \left(
W_{2}+W_{1}\right)  ^{-1}\left(  W_{1}-I\right)  \right]  \right)  \det\left(
i\Omega/2\right)  \right]  ^{1/2}\\
&  =\left[  \det\left(  -\left(  W_{2}-I\right)  \left(  W_{2}+W_{1}\right)
^{-1}\left(  W_{1}-I\right)  \right)  \det\left(  i\Omega/2\right)  \right]
^{1/2}\\
&  =\left[  \det\left(  \left(  I-W_{2}\right)  \left(  i\Omega/2\right)
\left(  -2i\Omega\right)  \left(  W_{2}+W_{1}\right)  ^{-1}\left(
I-W_{1}\right)  i\Omega/2\right)  \right]  ^{1/2}\\
&  =\left[  \det\left(  \left(  I-W_{2}\right)  \left(  i\Omega/2\right)
\right)  \det(\left(  -2i\Omega\right)  \left(  W_{2}+W_{1}\right)  ^{-1}%
)\det(\left(  I-W_{1}\right)  i\Omega/2)\right]  ^{1/2}\\
&  =\left[  \det\left(  \left[  V_{2}+i\Omega\right]  /2\right)  \det(\left(
-2i\Omega\right)  \left(  -V_{2}i\Omega-V_{1}i\Omega\right)  ^{-1})\det\left(
\left[  V_{1}+i\Omega\right]  /2\right)  \right]  ^{1/2}\\
&  =\left[  \det\left(  \left[  V_{2}+i\Omega\right]  /2\right)  \det(\left(
\left[  V_{2}+V_{1}\right]  /2\right)  ^{-1})\det\left(  \left[  V_{1}%
+i\Omega\right]  /2\right)  \right]  ^{1/2}\\
&  =\left[  \frac{\det\left(  \left[  V_{2}+i\Omega\right]  /2\right)
\det\left(  \left[  V_{1}+i\Omega\right]  /2\right)  }{\det(\left[
V_{2}+V_{1}\right]  /2)}\right]  ^{1/2}.
\end{align}
This concludes the proof.
\end{proof}

\begin{proposition}
\label{prop:inverse-sandwiched-by-sqrt}Given positive-definite real matrices
$H_{4}$ and $H_{5}$, we have that%
\begin{equation}
\exp\left[  -\frac{1}{2}\hat{x}^{T}\left[  H_{4}/2\right]  \hat{x}\right]
\exp\left[  -\frac{1}{2}\hat{x}^{T}\left[  -H_{5}\right]  \hat{x}\right]
\exp\left[  -\frac{1}{2}\hat{x}^{T}\left[  H_{4}/2\right]  \hat{x}\right]
=\exp\left[  -\frac{1}{2}\hat{x}^{T}H_{8}\hat{x}\right]  .
\end{equation}
In the above, $H_{8}$ is real and positive definite if $V_{5}>V_{4}$ and is
such that%
\begin{align}
\operatorname{Tr}\left\{  \exp\left[  -\frac{1}{2}\hat{x}^{T}H_{8}\hat
{x}\right]  \right\}   &  =\sqrt{\det(\left[  V_{8}+i\Omega\right]
/2)}\label{eq:H8-normalization}\\
&  =\sqrt{\frac{\det(\left[  V_{4}+i\Omega\right]  /2)\det(\left[
V_{5}+i\Omega\right]  /2)}{\det(\left[  V_{5}-V_{4}\right]  /2)}%
},\label{eq:alt-normalization}\\
V_{8}  &  =V_{4}+\left(  \sqrt{I+\left(  V_{4}\Omega\right)  ^{-2}}\right)
V_{4}\left(  V_{5}-V_{4}\right)  ^{-1}V_{4}\left(  \sqrt{I+\left(  \Omega
V_{4}\right)  ^{-2}}\right)  .
\end{align}

\end{proposition}

\begin{proof}
The proof of this proposition amounts to examining again the proofs of
Propositions~\ref{prop:product-rule-Gaussian}\ and
\ref{prop:sandwich-with-square-root} and instead substituting $-H_{5}$ for
$H_{5}$. Consider that the product rule from
Proposition~\ref{prop:product-rule-Gaussian} holds for symmetric $H_{1}$ and
$H_{2}$%
\begin{equation}
\exp\left[  -\frac{1}{2}\hat{x}^{T}H_{1}\hat{x}\right]  \exp\left[  -\frac
{1}{2}\hat{x}^{T}H_{2}\hat{x}\right]  =\exp\left[  -\frac{1}{2}\hat{x}%
^{T}H_{3}\hat{x}\right]  .
\end{equation}
From Ref.~\onlinecite{Balian1969}, we know that the symmetric matrix $H_{3}$ that
satisfies \eqref{eq:product-Gaussian-forms} is the same one that satisfies the
following equation:%
\begin{equation}
\exp\left[  -i\Omega H_{1}\right]  \exp\left[  -i\Omega H_{2}\right]
=\exp\left[  -i\Omega H_{3}\right]  .
\end{equation}
Note that, by taking inverses, this latter equation is equivalent to%
\begin{equation}
\exp\left[  i\Omega H_{3}\right]  =\exp\left[  i\Omega H_{2}\right]
\exp\left[  i\Omega H_{1}\right]  .
\end{equation}
Recalling that $\exp\left(  i\Omega H\right)  =\frac{W-I}{W+I}$, we find that%
\begin{equation}
\exp\left(  -i\Omega H\right)  =\left[  \exp\left(  i\Omega H\right)  \right]
^{-1}=\left[  \frac{W-I}{W+I}\right]  ^{-1}=\frac{W+I}{W-I}=\frac{-\left[
-W-I\right]  }{W-I}=\frac{-W-I}{-W+I}, \label{eq:minus-H-to-minus-W}%
\end{equation}
which implies that the transformation $H\rightarrow-H$ induces the
transformation $W\rightarrow-W$, as observed in Ref.~\onlinecite{LDW17}. Now propagating
this minus sign throughout all of the calculations in the proofs of
Propositions~\ref{prop:product-rule-Gaussian}\ and
\ref{prop:sandwich-with-square-root}, we find that%
\begin{align}
V_{8}  &  =V_{4}-\left(  \sqrt{I+\left(  V_{4}\Omega\right)  ^{-2}}\right)
V_{4}\left(  -V_{5}+V_{4}\right)  ^{-1}V_{4}\left(  \sqrt{I+\left(  \Omega
V_{4}\right)  ^{-2}}\right) \\
&  =V_{4}+\left(  \sqrt{I+\left(  V_{4}\Omega\right)  ^{-2}}\right)
V_{4}\left(  V_{5}-V_{4}\right)  ^{-1}V_{4}\left(  \sqrt{I+\left(  \Omega
V_{4}\right)  ^{-2}}\right)  .
\end{align}
The latter is a legitimate covariance matrix when $V_{5}-V_{4}>0$ because%
\begin{align}
V_{8}+i\Omega &  =V_{4}+i\Omega+\left(  \sqrt{I+\left(  V_{4}\Omega\right)
^{-2}}\right)  V_{4}\left(  V_{5}-V_{4}\right)  ^{-1}V_{4}\left(
\sqrt{I+\left(  \Omega V_{4}\right)  ^{-2}}\right) \\
&  \geq\left(  \sqrt{I+\left(  V_{4}\Omega\right)  ^{-2}}\right)  V_{4}\left(
V_{5}-V_{4}\right)  ^{-1}V_{4}\left(  \sqrt{I+\left(  \Omega V_{4}\right)
^{-2}}\right) \\
&  >0.
\end{align}
In the above, we used the fact that $V_{4}$ is a legitimate covariance matrix
satisfying $V_{4}+i\Omega\geq0$ and the assumption that $V_{5}-V_{4}>0$. By
\eqref{eq:H-PD-V-legit}, this implies that $H_{8}>0$. Since $V_{8}$ is a
legitimate covariance matrix corresponding to $H_{8}$, we conclude \eqref{eq:H8-normalization}.

A proof for \eqref{eq:alt-normalization} follows by examining again the proof
of Proposition~\ref{prop:simple-product-normalization} and considering again
that the transformation $H\rightarrow-H$ induces the transformation
$W\rightarrow-W$. Propagating the minus sign throughout the calculation, we
arrive at%
\begin{align}
&  \left[  \frac{\det\left(  \left[  -V_{2}+i\Omega\right]  /2\right)
\det\left(  \left[  V_{1}+i\Omega\right]  /2\right)  }{\det(\left[
-V_{2}+V_{1}\right]  /2)}\right]  ^{1/2}\nonumber\\
&  =\left[  \frac{\det\left(  \left[  V_{2}-i\Omega\right]  /2\right)
\det\left(  \left[  V_{1}+i\Omega\right]  /2\right)  }{\det(\left[
V_{2}-V_{1}\right]  /2)}\right]  ^{1/2}\\
&  =\left[  \frac{\det\left(  \left[  V_{2}+i\Omega\right]  /2\right)
\det\left(  \left[  V_{1}+i\Omega\right]  /2\right)  }{\det(\left[
V_{2}-V_{1}\right]  /2)}\right]  ^{1/2},
\end{align}
where the last equality follows because $\left[  V_{2}-i\Omega\right]
^{T}=V_{2}+i\Omega$ and the determinant is invariant with respect to transposition.
\end{proof}

We close this section by remarking that many of the above calculations can be completed by considering the approach developed in Ref.~\onlinecite[Appendix~A]{LDW17}.

\subsection{Mean vectors and displacement operators}

We begin by recalling some standard properties of the operator in
\eqref{eq: displacement Op}. For detailed proofs, see, e.g., Ref.~\onlinecite{S17}, but
note that they follow from the Baker--Campbell--Hausdorff formula and its corollaries.

\begin{proposition}
\label{prop: displ op properties}The displacement operator in
\eqref{eq: displacement Op} (extended to $s\in\mathbb{C}^{2n}$) satisfies the
following properties:
\begin{align}
D(s)^{-1}  &  =D(-s),\label{eq: dagger is neg}\\
D(s)D(t)  &  =D(s+t)e^{-\frac{1}{2}s^{T}i\Omega t}, \label{eq: prod of displs}%
\\
D(s)\hat{x}D(-s)  &  =\hat{x}+s,\label{eq: displ effect on x}\\
\exp\left[  -\frac{1}{2}\hat{x}^{T}H\hat{x}\right]  \hat{{x}}  &  =\exp\left[
i\Omega H\right]  \hat{x}\exp\left[  -\frac{1}{2}\hat{x}^{T}H\hat{x}\right]  ,
\label{eq: Gaussian sandwich of x}%
\end{align}
where $s,t\in\mathbb{C}^{2n}$ and $H$ is a symmetric matrix with complex entries. If
$s\in\mathbb{R}^{2n}$, then $D(s)^{-1}=D(s)^{\dag}$.
\end{proposition}

The following corollary generalizes some statements from Ref.~\onlinecite{BBP15}:

\begin{corollary}
\label{cor: displ op}The following equalities involving the displacement
operator and exponential quadratic forms hold%
\begin{equation}
\exp\left[  -\frac{1}{2}\hat{x}^{T}H\hat{x}\right]  D(s)=D\left(  \exp\left[
-i\Omega H\right]  s\right)  \exp\left[  -\frac{1}{2}\hat{x}^{T}H\hat
{x}\right]  , \label{eq: displ gauss 1}%
\end{equation}
and for $l=\left(  \exp\left[  -i\Omega H\right]  -I\right)  s,$
\begin{equation}
D(l)\exp\left[  -\frac{1}{2}\hat{x}^{T}H\hat{x}\right]  =D(-s)\exp\left[
-\frac{1}{2}\hat{x}^{T}H\hat{x}\right]  D(s)e^{\frac{1}{4}l^{T}i\Omega Wl},
\label{eq: displ gauss 2}%
\end{equation}
where $s\in\mathbb{C}^{2n}$, $H$ is a symmetric matrix with complex entries, and $W$ is
related to $H$ by \eqref{eq:H-to-W}.
\end{corollary}

\begin{proof}
Consider that%
\begin{equation}
\exp\left[  -\frac{1}{2}\hat{x}^{T}H\hat{x}\right]  \exp\left[  s^{T}%
i\Omega\hat{x}\right]  
=\exp\left[  s^{T}i\Omega\exp\left[  i\Omega H\right]  \hat{x}\right]  
\exp\left[-  \frac{1}{2}\hat{x}^{T}H\hat{x}\right],
\end{equation}
which follows from applying \eqref{eq: Gaussian sandwich of x} of Proposition
\ref{prop: displ op properties} to $\exp\left[  s^{T}i\Omega\hat{x}\right]  $.
This implies that%
\begin{align}
\exp\left[  -\frac{1}{2}\hat{x}^{T}H\hat{x}\right]  \exp\left[  s^{T}%
i\Omega\hat{x}\right]   &  =\exp\left[  s^{T}i\Omega\exp\left[  i\Omega
H\right]  \hat{x}\right]  \exp\left[  -\frac{1}{2}\hat{x}^{T}H\hat{x}\right]
\\
&  =\exp\left[  s^{T}\exp\left[  Hi\Omega\right]  i\Omega\hat{x}\right]
\exp\left[  -\frac{1}{2}\hat{x}^{T}H\hat{x}\right] \\
&  =\exp\left[  \left(  \exp\left[  -i\Omega H\right]  s\right)  ^{T}%
i\Omega\hat{x}\right]  \exp\left[  -\frac{1}{2}\hat{x}^{T}H\hat{x}\right]  ,
\end{align}
where we have used the following:%
\begin{equation}
i\Omega\exp\left[  i\Omega H\right]  =i\Omega\exp\left[  i\Omega H\right]
i\Omega i\Omega=\exp\left[  i\Omega i\Omega Hi\Omega\right]  i\Omega
=\exp\left[  Hi\Omega\right]  i\Omega.
\end{equation}
This establishes \eqref{eq: displ gauss 1}.

To see \eqref{eq: displ gauss 2}, consider that
\begin{align}
D(-s)\exp\left[  -\frac{1}{2}\hat{x}^{T}H\hat{x}\right]  D(s)  &
=D(-s)D(\exp\left[  -i\Omega H\right]  s)\exp\left[  -\frac{1}{2}\hat{x}%
^{T}H\hat{x}\right] \\
&  =e^{\frac{1}{2}s^{T}i\Omega\exp\left[  -i\Omega H\right]  s}D(l)\exp\left[
-\frac{1}{2}\hat{x}^{T}H\hat{x}\right]  ,
\end{align}
where
\begin{equation}
l=\left(  \exp\left[  -i\Omega H\right]  -I\right)  s.
\end{equation}
In the above, the first equality follows from applying
\eqref{eq: displ gauss 1}, while the second equality results from applying
\eqref{eq: prod of displs} of Proposition \ref{prop: displ op properties}.
Thus,
\begin{equation}
D(l)\exp\left[  -\frac{1}{2}\hat{x}^{T}H\hat{x}\right]  =D(-s)\exp\left[
-\frac{1}{2}\hat{x}^{T}H\hat{x}\right]  D(s)e^{-\frac{1}{2}s^{T}i\Omega
\exp\left[  -i\Omega H\right]  s}. \label{eq: towards displ gauss 2}%
\end{equation}
Furthermore, consider that
\begin{equation}
s^{T}i\Omega\exp\left[  -i\Omega H\right]  s=\frac{1}{2}s^{T}i\Omega\left(
\exp\left[  -i\Omega H\right]  -\exp\left[  i\Omega H\right]  \right)  s,
\label{eq:  s phase s real}%
\end{equation}
which follows because a scalar is equal to its transpose. Thus,
\begin{align}
&  \frac{1}{2}s^{T}i\Omega\exp\left[  -i\Omega H\right]  s\nonumber\\
&  =\frac{1}{4}s^{T}i\Omega\left(  \exp\left[  -i\Omega H\right]  -\exp\left[
i\Omega H\right]  \right)  s\label{eq: expression K beginning}\\
&  =\frac{1}{4}\left(  \left(  \exp\left[  -i\Omega H\right]  -I\right)
^{-1}l\right)  ^{T}i\Omega\left(  \exp\left[  -i\Omega H\right]  -\exp\left[
i\Omega H\right]  \right)  \left(  \exp\left[  -i\Omega H\right]  -I\right)
^{-1}l\\
&  =\frac{1}{4}l^{T}\left(  \exp\left[  Hi\Omega\right]  -I\right)
^{-1}i\Omega\left(  \exp\left[  -i\Omega H\right]  -\exp\left[  i\Omega
H\right]  \right)  \left(  \exp\left[  -i\Omega H\right]  -I\right)  ^{-1}l\\
&  =\frac{1}{4}l^{T}i\Omega\left(  \exp\left[  i\Omega H\right]  -I\right)
^{-1}\left(  \exp\left[  -i\Omega H\right]  -\exp\left[  i\Omega H\right]
\right)  \left(  \exp\left[  -i\Omega H\right]  -I\right)  ^{-1}l.
\end{align}
Now that we have expressed the middle operator in terms of the scalar function%
\begin{equation}
x\rightarrow\left[  x-1\right]  ^{-1}\left(  x^{-1}-x\right)  \left[
x^{-1}-1\right]  ^{-1}=\frac{x+1}{x-1},
\end{equation}
we can conclude that%
\begin{align}
\frac{1}{2}s^{T}i\Omega\exp\left[  -i\Omega H\right]  s  &  =\frac{1}{4}%
l^{T}i\Omega\frac{\left(  \exp\left[  i\Omega H\right]  +I\right)  }{\left(
\exp\left[  i\Omega H\right]  -I\right)  }l\\
&  =-\frac{1}{4}l^{T}i\Omega Wl. \label{eq: expression K}%
\end{align}
Equations \eqref{eq: towards displ gauss 2} and \eqref{eq: expression K}
together establish \eqref{eq: displ gauss 2}.
\end{proof}

\begin{remark}
\label{rem: Corollary displ Gaussian inverse}In Corollary \ref{cor: displ op},
if $H\rightarrow-H$, (i.e., if the inverse of an exponential quadratic form is
considered), then the above statements change as
\begin{align}
\exp\left[  -\frac{1}{2}\hat{x}^{T}\left(  -H\right)  \hat{x}\right]  D(s)  &
=D\left(  \exp\left[  -i\Omega\left(  -H\right)  \right]  s\right)
\exp\left[  -\frac{1}{2}\hat{x}^{T}\left(  -H\right)  \hat{x}\right]  ,\\
D(l)\exp\left[  -\frac{1}{2}\hat{x}^{T}\left(  -H\right)  \hat{x}\right]   &
=D(-s)\exp\left[  -\frac{1}{2}\hat{x}^{T}\left(  -H\right)  \hat{x}\right]
D(s)e^{\frac{1}{4}l^{T}i\Omega\left(  -W\right)  l},
\end{align}
for $l=\left(  \exp\left[  -i\Omega\left(  -H\right)  \right]  -I\right)  s,$
where $s\in\mathbb{C}^{2n}$, and $W$ is
related to $H$ by \eqref{eq:H-to-W}.
\end{remark}

\begin{lemma}
[{[\onlinecite{BBP15}]}]\label{lem: displ term} For positive-definite real matrices
$H_{9}$ and $H_{10}$ such that
\begin{equation}
\exp\left[  -i\Omega H_{9}\right]  \exp\left[  -i\Omega H_{10}\right]
=\exp\left[  -i\Omega H_{11}\right]  ,
\end{equation}
and%
\begin{equation}
l=\left(  \exp\left[  -i\Omega H_{9}\right]  -I\right)  s,\qquad
s\in\mathbb{C}^{2n},
\end{equation}
the following equality holds%
\begin{equation}
-\frac{1}{4}l^{T}i\Omega W_{9}l+\frac{1}{4}l^{T}i\Omega W_{11}l=-s^{T}\left(
V_{9}+V_{10}\right)  ^{-1}s, \label{eq: lemma with means}%
\end{equation}
where $V_{9}$ and $V_{10}$ are related to $H_{9}$ and $H_{10},$ respectively,
by \eqref{eq:gibbs-to-cov}, $W_{9}$ is related to $H_{9}$ by
\eqref{eq:H-to-W}, and
\begin{equation}
W_{11}=\frac{I+\exp\left(  i\Omega H_{11}\right)  }{I-\exp\left(  i\Omega
H_{11}\right)  }=\frac{I+\exp\left(  i\Omega H_{10}\right)  \exp\left(
i\Omega H_{9}\right)  }{I-\exp\left(  i\Omega H_{10}\right)  \exp\left(
i\Omega H_{9}\right)  }.
\end{equation}

\end{lemma}

\begin{proof}
From \eqref{eq: expression K beginning}--\eqref{eq: expression K} of
Corollary~\ref{cor: displ op}, we have that
\begin{equation}
-\frac{1}{4}l^{T}i\Omega W_{9}l=\frac{1}{4}s^{T}i\Omega\left(  \exp\left[
-i\Omega H_{9}\right]  -\exp\left[  i\Omega H_{9}\right]  \right)  s.
\end{equation}
We also have that%
\begin{align}
\frac{1}{4}l^{T}i\Omega W_{11}l  &  =\frac{1}{4}s^{T}\left(  \exp\left[
-i\Omega H_{9}\right]  -I\right)  ^{T}i\Omega W_{11}\left(  \exp\left[
-i\Omega H_{9}\right]  -I\right)  s\\
&  =\frac{1}{4}s^{T}i\Omega\left(  \exp\left[  i\Omega H_{9}\right]
-I\right)  W_{11}\left(  \exp\left[  -i\Omega H_{9}\right]  -I\right)  s,
\end{align}
so that the total expression can be written as
\begin{multline}
-\frac{1}{4}l^{T}i\Omega W_{9}l+\frac{1}{4}l^{T}i\Omega W_{11}l\\
=\frac{1}{4}s^{T}i\Omega\left[  \left(  \exp\left[  -i\Omega H_{9}\right]
-\exp\left[  i\Omega H_{9}\right]  \right)  +\left(  \exp\left[  i\Omega
H_{9}\right]  -I\right)  W_{11}\left(  \exp\left[  -i\Omega H_{9}\right]
-I\right)  \right]  s, \label{eq: displ lemma full expression}%
\end{multline}
The following equalities hold by exploiting \eqref{eq:W-to-H}%
\begin{align}
\exp\left[  -i\Omega H_{9}\right]  -\exp\left[  i\Omega H_{9}\right]   &
=\frac{W_{9}+I}{W_{9}-I}-\frac{W_{9}-I}{W_{9}+I}=\frac{4W_{9}}{W_{9}^{2}%
-I},\label{eq: displ lemma term 1}\\
\left(  \exp\left[  i\Omega H_{9}\right]  -I\right)  W_{11}\left(  \exp\left[
-i\Omega H_{9}\right]  -I\right)   &  =-\frac{2}{W_{9}+I}W_{11}\frac{2}%
{W_{9}-I}. \label{eq: displ lemma term 2}%
\end{align}
From \eqref{eq:use-later-sqrt-sand-last}, we have that
\begin{equation}
W_{11}=W_{9}-\left(  W_{9}+I\right)  \left(  W_{9}+W_{10}\right)  ^{-1}\left(
W_{9}-I\right)  , \label{eq: W'' in terms of W and W'}%
\end{equation}
Using \eqref{eq: W'' in terms of W and W'}, we thus have that%
\begin{align}
&  -\frac{2}{W_{9}+I}W_{11}\frac{2}{W_{9}-I}\nonumber\\
&  =-\frac{2}{W_{9}+I}\left[  W_{9}-\left(  W_{9}+I\right)  \left(
W_{9}+W_{10}\right)  ^{-1}\left(  W_{9}-I\right)  \right]  \frac{2}{W_{9}-I}\\
&  =-\frac{2}{W_{9}+I}W_{9}\frac{2}{W_{9}-I}+\frac{2}{W_{9}+I}\left(
W_{9}+I\right)  \left(  W_{9}+W_{10}\right)  ^{-1}\left(  W_{9}-I\right)
\frac{2}{W_{9}-I}\\
&  =-\frac{4W_{9}}{W_{9}^{2}-I}+4\left(  W_{9}+W_{10}\right)  ^{-1}.
\label{eq: displ lemma term 2 simplified}%
\end{align}
Combining \eqref{eq: displ lemma term 1} and
\eqref{eq: displ lemma term 2 simplified}, we obtain that
\begin{multline}
\frac{1}{4}s^{T}i\Omega\left[  \left(  \exp\left[  -i\Omega H_{9}\right]
-\exp\left[  i\Omega H_{9}\right]  \right)  +\left(  \exp\left[  i\Omega
H_{9}\right]  -I\right)  W_{11}\left(  \exp\left[  -i\Omega H_{9}\right]
-I\right)  \right]  s\\
=s^{T}i\Omega\left(  W_{9}+W_{10}\right)  ^{-1}s=-s^{T}\left(  V_{9}%
+V_{10}\right)  ^{-1}s,
\end{multline}
which is the statement of the lemma.
\end{proof}

\begin{lemma}
\label{lem: displ term with inverse}Let $H_{9}$ and $H_{10}$ be
positive-definite real matrices such that
\begin{equation}
\exp\left[  -i\Omega\left(  -H_{9}\right)  \right]  \exp\left[  -i\Omega
H_{10}\right]  =\exp\left[  -i\Omega H_{11}\right]  ,
\end{equation}
with $H_{11}$ satisfying the above, and let%
\begin{equation}
l=\left(  \exp\left[  - i\Omega(-H_{9})\right]  -I\right)  s,\qquad
s\in\mathbb{C}^{2n}.
\end{equation}
If $V_{9}>V_{10}$, then the following equality holds%
\begin{equation}
-\frac{1}{4}l^{T}i\Omega(-W_{9})l+\frac{1}{4}l^{T}i\Omega W_{11}l=
s^{T}\left(  V_{9} - V_{10}\right)  ^{-1}s, \label{eq: lemma with means-1}%
\end{equation}
where $V_{9}$ and $V_{10}$ are related to $H_{9}$ and $H_{10},$ respectively,
by \eqref{eq:gibbs-to-cov}, $W_{9}$ is related to $H_{9}$ by
\eqref{eq:H-to-W}, and
\begin{equation}
W_{11}=\frac{I+\exp\left(  i\Omega H_{11}\right)  }{I-\exp\left(  i\Omega
H_{11}\right)  }=\frac{I+\exp\left(  i\Omega H_{10}\right)  \exp\left(
-i\Omega H_{9}\right)  }{I-\exp\left(  i\Omega H_{10}\right)  \exp\left(
-i\Omega H_{9}\right)  }.
\end{equation}

\end{lemma}

\begin{proof}
This amounts to reexamining the proof of Lemma~\ref{lem: displ term},
and noting that, from the discussion around \eqref{eq:minus-H-to-minus-W},
$W\rightarrow-W$ and $V\rightarrow-V$ when $H\rightarrow-H,$ i.e., under the
inverse of an exponential quadratic form. This implies that $-s^{T}\left(
V_{9}+V_{10}\right)  ^{-1}s\rightarrow-s^{T}\left(  -V_{9}+V_{10}\right)
^{-1}s=s^{T}\left(  V_{9}-V_{10}\right)  ^{-1}s$. The condition $V_{9}>V_{10}$
suffices to guarantee that the matrix $V_{9}-V_{10}$ is invertible, which is
used throughout the calculations in the proof of Lemma~\ref{lem: displ term}.
\end{proof}

\section{Petz--R\'{e}nyi relative entropy}

\label{sec:Petz-Renyi}

We now determine a formula for the Petz--R\'{e}nyi relative entropy
\cite{P86}\ of two Gaussian states $\rho$ and $\sigma$. The Petz--R\'{e}nyi
relative entropy is defined for $\alpha\in(0,1)\cup(1,\infty)$ as%
\begin{equation}
D_{\alpha}(\rho\Vert\sigma)\equiv\frac{1}{\alpha-1}\ln Q_{\alpha}(\rho
\Vert\sigma),
\end{equation}
where $Q_{\alpha}(\rho\Vert\sigma)$ denotes the Petz--R\'{e}nyi relative
quasi-entropy:%
\begin{equation}
Q_{\alpha}(\rho\Vert\sigma)\equiv\operatorname{Tr}\left\{  \rho^{\alpha}%
\sigma^{1-\alpha}\right\}  . \label{eq:Petz-quasi}%
\end{equation}
We first consider the case when $\alpha\in(0,1)$, and then we move on to the
case when $\alpha\in(1,\infty)$.

A formula for the Petz--R\'{e}nyi relative entropy was previously given in
Ref.~\onlinecite{PL08} for $\alpha\in(0,1)$. The formula given in
Theorem~\ref{thm:petz-renyi-alpha<1}\ below is expressed differently from
the one given in Ref.~\onlinecite{PL08} because it depends directly on the covariance
matrices of the states involved.

\begin{theorem}
\label{thm:petz-renyi-alpha<1}
Let $\alpha\in(0,1)$, and let $\rho$ and $\sigma$ denote two
Gaussian states.
Then the Petz--R\'{e}nyi relative quasi-entropy
$Q_{\alpha}(\rho\Vert\sigma)$ as defined in \eqref{eq:Petz-quasi}  is given by%
\begin{equation}
Q_{\alpha}(\rho\Vert\sigma)=\frac{Z_{\rho(\alpha)}Z_{\sigma(1-\alpha)}%
}{Z_{\rho}^{\alpha}Z_{\sigma}^{1-\alpha}\left[  \det\left(  \left[
V_{\rho(\alpha)}+V_{\sigma(1-\alpha)}\right]  /2\right)  \right]  ^{1/2}}%
\exp\left\{  -\delta s^{T}\left(  V_{\rho(\alpha)}+V_{\sigma(1-\alpha
)}\right)  ^{-1}\delta s\right\}  , \label{eq:petz-renyi-alpha<1}%
\end{equation}
where%
\begin{align}
V_{\rho(\alpha)}  &  =\frac{\left(  I+\left(  V_{\rho}i\Omega\right)
^{-1}\right)  ^{\alpha}+\left(  I-\left(  V_{\rho}i\Omega\right)
^{-1}\right)  ^{\alpha}}{\left(  I+\left(  V_{\rho}i\Omega\right)
^{-1}\right)  ^{\alpha}-\left(  I-\left(  V_{\rho}i\Omega\right)
^{-1}\right)  ^{\alpha}}i\Omega,\label{eq:CM-rho-to-alpha}\\
V_{\sigma(1-\alpha)}  &  =\frac{\left(  I+\left(  V_{\sigma}i\Omega\right)
^{-1}\right)  ^{1-\alpha}+\left(  I-\left(  V_{\sigma}i\Omega\right)
^{-1}\right)  ^{1-\alpha}}{\left(  I+\left(  V_{\sigma}i\Omega\right)
^{-1}\right)  ^{1-\alpha}-\left(  I-\left(  V_{\sigma}i\Omega\right)
^{-1}\right)  ^{1-\alpha}}i\Omega,\label{eq:CM-sigma-to-1-alpha}\\
Z_{\rho(\alpha)}  &  =\sqrt{\det(\left[  V_{\rho(\alpha)}+i\Omega\right]
/2)},\\
Z_{\sigma(1-\alpha)}  &  =\sqrt{\det(\left[  V_{\sigma(1-\alpha)}%
+i\Omega\right]  /2)},\\
\delta s  &  =s_{\rho}-s_{\sigma}.
\end{align}

\end{theorem}

\begin{proof}
Let $\rho_{0}$ and $\sigma_{0}$ denote the following operators:%
\begin{equation}
\rho_{0}=\exp\left[  -\frac{1}{2}\hat{x}^{T}H_{\rho}\hat{x}\right]
,\qquad\sigma_{0}=\exp\left[  -\frac{1}{2}\hat{x}^{T}H_{\sigma}\hat{x}\right]
.
\end{equation}
Consider that%
\begin{align}
\operatorname{Tr}\{\rho^{\alpha}\sigma^{1-\alpha}\}  &  =\operatorname{Tr}%
\left\{  \left[  D(-s_{\rho})\left(  \frac{\rho_{0}}{Z_{\rho}}\right)
D(s_{\rho})\right]  ^{\alpha}\left[  D(-s_{\sigma})\left(  \frac{\sigma_{0}%
}{Z_{\sigma}}\right)  D(s_{\sigma})\right]  ^{1-\alpha}\right\} \\
&  =\frac{1}{Z_{\rho}^{\alpha}Z_{\sigma}^{1-\alpha}}\operatorname{Tr}\left\{
D(-s_{\rho})\left(  \rho_{0}\right)  ^{\alpha}D(s_{\rho})D(-s_{\sigma})\left(
\sigma_{0}\right)  ^{1-\alpha}D(s_{\sigma})\right\} \\
&  =\frac{1}{Z_{\rho}^{\alpha}Z_{\sigma}^{1-\alpha}}\operatorname{Tr}\left\{
D(-\delta s)\left(  \rho_{0}\right)  ^{\alpha}D\left(  \delta s\right)
\left(  \sigma_{0}\right)  ^{1-\alpha}\right\}  ,
\end{align}
where $\delta s=s_{\rho}-s_{\sigma}$. Using \eqref{eq: displ gauss 2} of
Corollary~\ref{cor: displ op} for $D(-\delta s)\left(  \rho_{0}\right)
^{\alpha}D\left(  \delta s\right)  $, we have that%
\begin{equation}
\left(  \frac{1}{Z_{\rho}^{\alpha}Z_{\sigma}^{1-\alpha}}\right)
^{-1}\operatorname{Tr}\{\rho^{\alpha}\sigma^{1-\alpha}\}=e^{-\frac{1}{4}%
l^{T}i\Omega W_{\rho\left(  \alpha\right)  }l}\operatorname{Tr}\{D(l)\left(
\rho_{0}\right)  ^{\alpha}\left(  \sigma_{0}\right)  ^{1-\alpha}\},
\end{equation}
where $W_{\rho\left(  \alpha\right)  }$ is related to $H_{\rho\left(
\alpha\right)  }=\alpha H_{\rho}$ by \eqref{eq:H-to-W}, and $l$ is given by
\begin{equation}
l=\left(  \exp\left[  -i\Omega\alpha H_{\rho}\right]  -I\right)  \delta s.
\end{equation}
Using \eqref{eq: displ gauss 2} of Corollary \ref{cor: displ op} once again,
we get that%
\begin{align}
\left(  \frac{1}{Z_{\rho}^{\alpha}Z_{\sigma}^{1-\alpha}}\right)
^{-1}\operatorname{Tr}\{\rho^{\alpha}\sigma^{1-\alpha}\}  &  =e^{-\frac{1}%
{4}l^{T}i\Omega W_{\rho\left(  \alpha\right)  }l}e^{\frac{1}{4}l^{T}i\Omega
W_{\xi}l}\operatorname{Tr}\{D(-t)\left(  \rho_{0}\right)  ^{\alpha}\left(
\sigma_{0}\right)  ^{1-\alpha}D\left(  t\right)  \}\\
&  =e^{-\frac{1}{4}l^{T}i\Omega W_{\rho\left(  \alpha\right)  }l}e^{\frac
{1}{4}l^{T}i\Omega W_{\xi}l}\operatorname{Tr}\{\left(  \rho_{0}\right)
^{\alpha}\left(  \sigma_{0}\right)  ^{1-\alpha}\},
\end{align}
where%
\begin{equation}
W_{\xi}=\frac{I+\exp\left[  i\Omega\left(  1-\alpha\right)  H_{\sigma}\right]
\exp\left[  i\Omega\alpha H_{\rho}\right]  }{I-\exp\left[  i\Omega\left(
1-\alpha\right)  H_{\sigma}\right]  \exp\left[  i\Omega\alpha H_{\rho}\right]
},
\end{equation}
and we have used
\begin{equation}
t=\left(  \exp\left[  -i\Omega\alpha H_{\rho}\right]  \exp\left[
-i\Omega\left(  1-\alpha\right)  H_{\sigma}\right]  -I\right)  ^{-1}\left(
\exp\left[  -i\Omega\alpha H_{\rho}\right]  -I\right)  \delta s.
\end{equation}
Note that the particular value of $t$ is irrelevant because the operators
$D\left(  t\right)  $ and $D(-t)$ cancel each other in the trace operation.
Applying Lemma~\ref{lem: displ term} with $H_{9}=\alpha H_{\rho}%
=H_{\rho\left(  \alpha\right)  }$ and $H_{10}=\left(  1-\alpha\right)
H_{\sigma}=H_{\sigma\left(  1-\alpha\right)  },$ we see that
\begin{equation}
e^{-\frac{1}{4}l^{T}i\Omega W_{\rho\left(  \alpha\right)  }l}e^{\frac{1}%
{4}l^{T}i\Omega W_{\xi}l}=\exp\left[  -\delta s^{T}\left(  V_{\rho\left(
\alpha\right)  }+V_{\sigma\left(  1-\alpha\right)  }\right)  ^{-1}\delta
s\right]  . \label{eq: Petz renyi prefactor a<1}%
\end{equation}

What remains now is to evaluate
\begin{equation}
\operatorname{Tr}\{\left(  \rho_{0}\right)  ^{\alpha}\left(  \sigma
_{0}\right)  ^{1-\alpha}\}=\operatorname{Tr}\left\{  \exp\left[  -\frac{1}%
{2}\hat{x}^{T}\left(  \alpha H_{\rho}\right)  \hat{x}\right]  \exp\left[
-\frac{1}{2}\hat{x}^{T}\left(  \left(  1-\alpha\right)  H_{\sigma}\right)
\hat{x}\right]  \right\}  .
\end{equation}
By Proposition~\ref{prop:rho-to-alpha}, the covariance matrix corresponding to
$\alpha H_{\rho}$ is given in \eqref{eq:CM-rho-to-alpha}, and the covariance
matrix corresponding to $\left(  1-\alpha\right)  H_{\sigma}$ is given in
\eqref{eq:CM-sigma-to-1-alpha}. We now apply
Proposition~\ref{prop:simple-product-normalization} to conclude that%
\begin{equation}
\operatorname{Tr}\left\{  \exp\left[  -\frac{1}{2}\hat{x}^{T}\left(  \alpha
H_{\rho}\right)  \hat{x}\right]  \exp\left[  -\frac{1}{2}\hat{x}^{T}\left(
\left(  1-\alpha\right)  H_{\sigma}\right)  \hat{x}\right]  \right\}
=\frac{Z_{\rho(\alpha)}Z_{\sigma(1-\alpha)}}{\left[  \det\left(  \left[
V_{\rho(\alpha)}+V_{\sigma(1-\alpha)}\right]  /2\right)  \right]  ^{1/2}}.
\end{equation}
This concludes the proof.
\end{proof}

\begin{theorem}
\label{prop:petz-renyi-alpha>1}
Let $\alpha\in(1,\infty)$, and let 
$\rho$ and $\sigma$ denote
two
Gaussian states such that
$V_{\sigma(\alpha-1)}>V_{\rho(\alpha)}$.
Then the Petz--R\'{e}nyi relative quasi-entropy
$Q_{\alpha}(\rho\Vert\sigma)$ as  defined in \eqref{eq:Petz-quasi} is given by%
\begin{equation}
Q_{\alpha}(\rho\Vert\sigma)=\frac{Z_{\sigma}^{\alpha-1}}{Z_{\rho}^{\alpha}%
}\frac{Z_{\rho(\alpha)}Z_{\sigma(\alpha-1)}}{\left[  \det\left(  \left[
V_{\sigma(\alpha-1)}-V_{\rho(\alpha)}\right]  /2\right)  \right]  ^{1/2}}%
\exp\left[  \delta s^{T}\left(  V_{\sigma\left(  \alpha-1\right)  }%
-V_{\rho\left(  \alpha\right)  }\right)  ^{-1}\delta s\right]  ,
\label{eq:petz-renyi-alpha>1}%
\end{equation}
where%
\begin{align}
V_{\rho(\alpha)}  &  =\frac{\left(  I+\left(  V_{\rho}i\Omega\right)
^{-1}\right)  ^{\alpha}+\left(  I-\left(  V_{\rho}i\Omega\right)
^{-1}\right)  ^{\alpha}}{\left(  I+\left(  V_{\rho}i\Omega\right)
^{-1}\right)  ^{\alpha}-\left(  I-\left(  V_{\rho}i\Omega\right)
^{-1}\right)  ^{\alpha}}i\Omega,\label{eq:CM-rho-alpha-10}\\
V_{\sigma(\alpha-1)}  &  =\frac{\left(  I+\left(  V_{\sigma}i\Omega\right)
^{-1}\right)  ^{\alpha-1}+\left(  I-\left(  V_{\sigma}i\Omega\right)
^{-1}\right)  ^{\alpha-1}}{\left(  I+\left(  V_{\sigma}i\Omega\right)
^{-1}\right)  ^{\alpha-1}-\left(  I-\left(  V_{\sigma}i\Omega\right)
^{-1}\right)  ^{\alpha-1}}i\Omega,\label{eq:CM-sigma-alpha-1-10}\\
Z_{\rho(\alpha)}  &  =\sqrt{\det(\left[  V_{\rho(\alpha)}+i\Omega\right]
/2)},\\
Z_{\sigma(\alpha-1)}  &  =\sqrt{\det(\left[  V_{\sigma(\alpha-1)}%
+i\Omega\right]  /2)},\\
\delta s  &  =s_{\rho}-s_{\sigma}.
\end{align}

\end{theorem}

\begin{proof}
Let $\rho_{0}$ and $\sigma_{0}$ denote the following operators:%
\begin{equation}
\rho_{0}=\exp\left[  -\frac{1}{2}\hat{x}^{T}H_{\rho}\hat{x}\right]
,\qquad\sigma_{0}=\exp\left[  -\frac{1}{2}\hat{x}^{T}H_{\sigma}\hat{x}\right]
.
\end{equation}
Consider that%
\begin{align}
\operatorname{Tr}\{\rho^{\alpha}\sigma^{1-\alpha}\}  &  =\operatorname{Tr}%
\{\rho^{\alpha}\left[  \sigma^{\alpha-1}\right]  ^{-1}\}=\operatorname{Tr}%
\{\left[  \sigma^{\alpha-1}\right]  ^{-1}\rho^{\alpha}\}\\
&  =\operatorname{Tr}\left\{  \left[  \left[  D(-s_{\sigma})\left(
\frac{\sigma_{0}}{Z_{\sigma}}\right)  D(s_{\sigma})\right]  ^{\alpha
-1}\right]  ^{-1}\left[  D(-s_{\rho})\left(  \frac{\rho_{0}}{Z_{\rho}}\right)
D(s_{\rho})\right]  ^{\alpha}\right\} \\
&  =\frac{Z_{\sigma}^{\alpha-1}}{Z_{\rho}^{\alpha}}\operatorname{Tr}\left\{
D(-s_{\sigma})\left[  \left(  \sigma_{0}\right)  ^{\alpha-1}\right]
^{-1}D(s_{\sigma})D(-s_{\rho})\left(  \rho_{0}\right)  ^{\alpha}D(s_{\rho
})\right\} \\
&  =\frac{Z_{\sigma}^{\alpha-1}}{Z_{\rho}^{\alpha}}\operatorname{Tr}\left\{
D(\delta s)\left[  \left(  \sigma_{0}\right)  ^{\alpha-1}\right]
^{-1}D\left(  -\delta s\right)  \left(  \rho_{0}\right)  ^{\alpha}\right\}  .
\end{align}
Using steps similar to those in the proof of
Theorem~\ref{thm:petz-renyi-alpha<1}, and based on
Remark~\ref{rem: Corollary displ Gaussian inverse}, we arrive at
\begin{equation}
\left(  \frac{Z_{\sigma}^{\alpha-1}}{Z_{\rho}^{\alpha}}\right)  ^{-1}%
\operatorname{Tr}\{\rho^{\alpha}\sigma^{1-\alpha}\}=e^{\frac{1}{4}l^{T}i\Omega
W_{\sigma\left(  \alpha-1\right)  }l}e^{\frac{1}{4}l^{T}i\Omega W_{\xi
^{\prime}}l}\operatorname{Tr}\left\{  \left(  \rho_{0}\right)  ^{\alpha
}\left[  \left(  \sigma_{0}\right)  ^{\alpha-1}\right]  ^{-1}\right\}  ,
\end{equation}
where
\begin{align}
W_{\xi^{\prime}}  &  =\frac{I+\exp\left[  i\Omega\alpha H_{\rho}\right]
\exp\left[  -i\Omega\left(  \alpha-1\right)  H_{\sigma}\right]  }%
{I-\exp\left[  i\Omega\alpha H_{\rho}\right]  \exp\left[  -i\Omega\left(
\alpha-1\right)  H_{\sigma}\right]  },\\
l  &  =\left(  \exp\left[  i\Omega\left(  \alpha-1\right)  H_{\sigma}\right]
-I\right)  \left(  -\delta s\right)  ,
\end{align}
and $W_{\sigma\left(  \alpha-1\right)  }$ is related to the operator
$H_{\sigma\left(  \alpha-1\right)  }=\left(  \alpha-1\right)  H_{\sigma}$ by
\eqref{eq:H-to-W}. Using Lemma~\ref{lem: displ term with inverse} with
$H_{9}=\left(  \alpha-1\right)  H_{\sigma}=H_{\sigma\left(  \alpha-1\right)
}$ and $H_{10}=\alpha H_{\rho}=H_{\rho\left(  \alpha\right)  }$, we then have
that
\begin{equation}
e^{\frac{1}{4}l^{T}i\Omega W_{\sigma\left(  \alpha-1\right)  }l}e^{\frac{1}%
{4}l^{T}i\Omega W_{\xi^{\prime}}l}=\exp\left[  \delta s^{T}\left(
V_{\sigma\left(  \alpha-1\right)  }-V_{\rho\left(  \alpha\right)  }\right)
^{-1}\delta s\right]  .
\end{equation}

To finish the proof, consider that%
\begin{align}
&  \operatorname{Tr}\{\rho_{0}^{\alpha}\sigma_{0}^{1-\alpha}\}\nonumber\\
&  =\operatorname{Tr}\{\rho_{0}^{\alpha/2}\left[  \sigma_{0}^{\alpha
-1}\right]  ^{-1}\rho_{0}^{\alpha/2}\}\\
&  =\operatorname{Tr}\left\{  \left[  \exp\left[  -\frac{1}{2}\hat{x}%
^{T}H_{\rho}\hat{x}\right]  \right]  ^{\alpha/2}\left[  \left[  \exp\left[
-\frac{1}{2}\hat{x}^{T}\left(  H_{\sigma}\right)  \hat{x}\right]  \right]
^{\alpha-1}\right]  ^{-1}\left[  \exp\left[  -\frac{1}{2}\hat{x}^{T}H_{\rho
}\hat{x}\right]  \right]  ^{\alpha/2}\right\} \\
&  =\operatorname{Tr}\left\{  \left[  \exp\left[  -\frac{1}{2}\hat{x}%
^{T}\left(  \alpha H_{\rho}\right)  \hat{x}\right]  \right]  ^{1/2}\left[
\exp\left[  -\frac{1}{2}\hat{x}^{T}\left(  \left(  \alpha-1\right)  H_{\sigma
}\right)  \hat{x}\right]  \right]  ^{-1}\left[  \exp\left[  -\frac{1}{2}%
\hat{x}^{T}\left(  \alpha H_{\rho}\right)  \hat{x}\right]  \right]
^{1/2}\right\}  .
\end{align}
By Proposition~\ref{prop:rho-to-alpha}, the covariance matrix corresponding to
$\alpha H_{\rho}$ is given by \eqref{eq:CM-rho-alpha-10}, and the covariance
matrix for $\left(  \alpha-1\right)  H_{\sigma}$ is given by
\eqref{eq:CM-sigma-alpha-1-10}. We can then apply
Proposition~\ref{prop:inverse-sandwiched-by-sqrt} to find that%
\begin{multline}
\operatorname{Tr}\left\{  \left[  \exp\left[  -\frac{1}{2}\hat{x}^{T}\left(
\alpha H_{\rho}\right)  \hat{x}\right]  \right]  ^{1/2}\left[  \exp\left[
-\frac{1}{2}\hat{x}^{T}\left(  \left(  \alpha-1\right)  H_{\sigma}\right)
\hat{x}\right]  \right]  ^{-1}\left[  \exp\left[  -\frac{1}{2}\hat{x}%
^{T}\left(  \alpha H_{\rho}\right)  \hat{x}\right]  \right]  ^{1/2}\right\} \\
=\frac{Z_{\rho(\alpha)}Z_{\sigma(\alpha-1)}}{\left[  \det\left(  \left[
V_{\sigma(\alpha-1)}-V_{\rho(\alpha)}\right]  /2\right)  \right]  ^{1/2}}.
\end{multline}
For this equality to hold, it suffices that $V_{\sigma(\alpha-1)}%
-V_{\rho(\alpha)}>0$, as discussed in the proof of
Proposition~\ref{prop:inverse-sandwiched-by-sqrt}. Putting everything
together, we arrive at \eqref{eq:petz-renyi-alpha>1}.
\end{proof}

\bigskip

The quasi-entropy $\operatorname{Tr}\{\rho^{2}\sigma^{-1}\}$ has a number of
applications that are discussed in Section~\ref{sec:covert-app}. As a special
case of Theorem~\ref{prop:petz-renyi-alpha>1}, we arrive at the following
expression for the quasi-entropy $\operatorname{Tr}\{\rho^{2}\sigma^{-1}\}$
after applying Corollary~\ref{cor:rho-to-the-2}:

\begin{corollary}
Let $\alpha=2$, and let $\rho$ and $\sigma$ denote two Gaussian states such that
$V_{\sigma}>V_{\rho(2)}$.
Then the Petz--R\'{e}nyi relative quasi-entropy as defined in \eqref{eq:Petz-quasi}
 is given by
\begin{align}
Q_{2}(\rho\Vert\sigma)  &  =\operatorname{Tr}\{\rho^{2}\sigma^{-1}\}\\
&  =\frac{Z_{\sigma}^{2}}{Z_{\rho}^{2}}\frac{Z_{\rho(2)}}{\left[  \det\left(
\left[  V_{\sigma}-V_{\rho(2)}\right]  /2\right)  \right]  ^{1/2}}\exp\left[
\delta s^{T}\left(  V_{\sigma}-V_{\rho\left(  2\right)  }\right)  ^{-1}\delta
s\right]  ,
\end{align}
where%
\begin{align}
V_{\rho(2)}  &  =\frac{1}{2}\left(  V_{\rho}+\Omega V_{\rho}^{-1}\Omega
^{T}\right)  ,\\
Z_{\rho(2)}  &  =\sqrt{\det(\left[  V_{\rho(2)}+i\Omega\right]  /2)},\\
Z_{\sigma}  &  =\sqrt{\det(\left[  V_{\sigma}+i\Omega\right]  /2)}.
\end{align}

\end{corollary}

\section{Sandwiched R\'{e}nyi relative entropy}

\label{sec:sandwiched-Renyi}

We now determine a formula for the sandwiched R\'{e}nyi relative entropy
\cite{MDSFT13,WWY14}\ of two quantum Gaussian states $\rho$ and $\sigma$. The
sandwiched R\'{e}nyi relative entropy is defined for $\alpha\in(0,1)\cup
(1,\infty)$ as%
\begin{equation}
\widetilde{D}_{\alpha}(\rho\Vert\sigma)\equiv\frac{1}{\alpha-1}\ln
\widetilde{Q}_{\alpha}(\rho\Vert\sigma),
\end{equation}
where $\widetilde{Q}_{\alpha}(\rho\Vert\sigma)$ denotes the sandwiched
R\'{e}nyi relative quasi-entropy:%
\begin{align}
\widetilde{Q}_{\alpha}(\rho\Vert\sigma)  &  \equiv\operatorname{Tr}\left\{
\left(  \sigma^{\left(  1-\alpha\right)  /2\alpha}\rho\sigma^{\left(
1-\alpha\right)  /2\alpha}\right)  ^{\alpha}\right\}  ,\\
&  =\operatorname{Tr}\left\{  \left(  \rho^{1/2}\sigma^{\left(  1-\alpha
\right)  /\alpha}\rho^{1/2}\right)  ^{\alpha}\right\}  .
\label{eq:sandwiched-quasi}%
\end{align}
The second equality follows because the eigenvalues of 
$  \sigma^{\left(  1-\alpha\right)  /2\alpha}\rho\sigma^{\left(
1-\alpha\right)  /2\alpha}$ and 
$\rho^{1/2}\sigma^{\left(  1-\alpha
\right)  /\alpha}\rho^{1/2}$ are equal, but the latter expression is easier to work with, and thus we do so in what follows.

\begin{theorem}
\label{thm: Sandwiched RRE a<1}
Let $\alpha\in(0,1)$, and let
$\rho$ and $\sigma$ denote two Gaussian states.
Then the sandwiched R\'{e}nyi relative quasi-entropy
$\widetilde{Q}_{\alpha}(\rho\Vert\sigma)$ as defined in
\eqref{eq:sandwiched-quasi}
 is given by%
\begin{equation}
\widetilde{Q}_{\alpha}(\rho\Vert\sigma)=\frac{1}{Z_{\rho}^{\alpha}Z_{\sigma
}^{1-\alpha}}\left[  \det\left(  \left[  V_{\xi(\alpha)}+i\Omega\right]
/2\right)  \right]  ^{1/2}\exp\left\{  -\alpha\ \delta s^{T}\left(
V_{\sigma\left(  \beta\right)  }+V_{\rho}\right)  ^{-1}\delta s\right\}  ,
\end{equation}
where%
\begin{align}
V_{\xi(\alpha)}  &  =\frac{\left(  I+\left(  V_{\xi}i\Omega\right)
^{-1}\right)  ^{\alpha}+\left(  I-\left(  V_{\xi}i\Omega\right)  ^{-1}\right)
^{\alpha}}{\left(  I+\left(  V_{\xi}i\Omega\right)  ^{-1}\right)  ^{\alpha
}-\left(  I-\left(  V_{\xi}i\Omega\right)  ^{-1}\right)  ^{\alpha}}%
i\Omega,\label{eq:sandwich-V-ksi-alpha}\\
V_{\xi}  &  =V_{\rho}-\sqrt{I+(V_{\rho}\Omega)^{-2}}V_{\rho}(V_{\sigma(\beta
)}+V_{\rho})^{-1}V_{\rho}\sqrt{I+(\Omega V_{\rho})^{-2}}%
,\label{eq:sandwich-intermediate-CM}\\
V_{\sigma(\beta)}  &  =\frac{\left(  I+\left(  V_{\sigma}i\Omega\right)
^{-1}\right)  ^{\beta}+\left(  I-\left(  V_{\sigma}i\Omega\right)
^{-1}\right)  ^{\beta}}{\left(  I+\left(  V_{\sigma}i\Omega\right)
^{-1}\right)  ^{\beta}-\left(  I-\left(  V_{\sigma}i\Omega\right)
^{-1}\right)  ^{\beta}}i\Omega,\label{eq:sigma-beta-sandwich-q}\\
\beta &  =\left(  1-\alpha\right)  /\alpha,\\
\delta s  &  =s_{\rho}-s_{\sigma}.
\end{align}

\end{theorem}

\begin{proof}
Let $\rho_{0}$ and $\sigma_{0}$ denote the following operators:
\begin{equation}
\rho_{0}=\exp\left[  -\frac{1}{2}\hat{x}^{T}H_{\rho}\hat{x}\right]
,\qquad\sigma_{0}=\exp\left[  -\frac{1}{2}\hat{x}^{T}H_{\sigma}\hat{x}\right]
.
\end{equation}
To evaluate the expression for the sandwiched R\'{e}nyi relative
quasi-entropy, consider that
\begin{align}
&  \operatorname{Tr}\left\{  \left(  \rho^{1/2}\sigma^{\left(  1-\alpha
\right)  /\alpha}\rho^{1/2}\right)  ^{\alpha}\right\} \nonumber\\
&  =\operatorname{Tr}\left\{  \left(  \rho^{1/2}\sigma^{\beta}\rho
^{1/2}\right)  ^{\alpha}\right\} \\
&  =\operatorname{Tr}\left\{  \left(  \left[  D(-s_{\rho})\left(  \frac
{\rho_{0}}{Z_{\rho}}\right)  D(s_{\rho})\right]  ^{\frac{1}{2}}\left[
D(-s_{\sigma})\left(  \frac{\sigma_{0}}{Z_{\sigma}}\right)  D(s_{\sigma
})\right]  ^{\beta}\left[  D(-s_{\rho})\left(  \frac{\rho_{0}}{Z_{\rho}%
}\right)  D(s_{\rho})\right]  ^{\frac{1}{2}}\right)  ^{\alpha}\right\} \\
&  =\operatorname{Tr}\left\{  \left(  D(-s_{\rho})\left(  \frac{\rho_{0}%
}{Z_{\rho}}\right)  ^{\frac{1}{2}}D\left(  \delta s\right)  \left(
\frac{\sigma_{0}}{Z_{\sigma}}\right)  ^{\beta}D\left(  -\delta s\right)
\left(  \frac{\rho_{0}}{Z_{\rho}}\right)  ^{\frac{1}{2}}D(s_{\rho})\right)
^{\alpha}\right\} \\
&  =\frac{1}{Z_{\rho}^{\alpha}Z_{\sigma}^{1-\alpha}}\operatorname{Tr}\left\{
D(-s_{\rho})\left(  \left(  \rho_{0}\right)  ^{\frac{1}{2}}D\left(  \delta
s\right)  \left(  \sigma_{0}\right)  ^{\beta}D\left(  -\delta s\right)
\left(  \rho_{0}\right)  ^{\frac{1}{2}}\right)  ^{\alpha}D(s_{\rho})\right\}
\\
&  =\frac{1}{Z_{\rho}^{\alpha}Z_{\sigma}^{1-\alpha}}\operatorname{Tr}\left\{
\left(  \left(  \rho_{0}\right)  ^{\frac{1}{2}}D\left(  \delta s\right)
\left(  \sigma_{0}\right)  ^{\beta}D\left(  -\delta s\right)  \left(  \rho
_{0}\right)  ^{\frac{1}{2}}\right)  ^{\alpha}\right\}  .
\label{eq:sandwiched-1st-block-lastline}%
\end{align}
Using \eqref{eq: displ gauss 2} of Corollary \ref{cor: displ op} for $D(\delta
s)\left(  \sigma_{0}\right)  ^{\beta}D\left(  -\delta s\right)  $, we obtain
\begin{equation}
\left(  \frac{1}{Z_{\rho}^{\alpha}Z_{\sigma}^{1-\alpha}}\right)
^{-1}\operatorname{Tr}\left\{  \left(  \rho^{1/2}\sigma^{\left(
1-\alpha\right)  /\alpha}\rho^{1/2}\right)  ^{\alpha}\right\}  =e^{-\frac
{\alpha}{4}l^{T}i\Omega W_{\sigma\left(  \beta\right)  }l}\operatorname{Tr}%
\Bigg\{\Bigg(\left(  \rho_{0}\right)  ^{\frac{1}{2}}D(l)\left(  \sigma
_{0}\right)  ^{\beta}\left(  \rho_{0}\right)  ^{\frac{1}{2}}\Bigg)^{\alpha
}\Bigg\},
\end{equation}
where $W_{\sigma\left(  \beta\right)  }$ is related to $H_{\sigma\left(
\beta\right)  }=\beta H_{\sigma}$ by \eqref{eq:H-to-W}, and $l$ is given by
\begin{equation}
l=\left(  \exp\left[  -i\Omega H_{\sigma\left(  \beta\right)  }\right]
-I\right)  \left(  -\delta s\right)  .
\end{equation}
Continuing further, using \eqref{eq: displ gauss 1} of Corollary
\ref{cor: displ op}, we get
\begin{multline}
\left(  \frac{1}{Z_{\rho}^{\alpha}Z_{\sigma}^{1-\alpha}}\right)
^{-1}\operatorname{Tr}\left\{  \left(  \rho^{1/2}\sigma^{\left(
1-\alpha\right)  /\alpha}\rho^{1/2}\right)  ^{\alpha}\right\} \\
=e^{-\frac{\alpha}{4}l^{T}i\Omega W_{\sigma\left(  \beta\right)  }%
l}\operatorname{Tr}\Bigg\{\Bigg(D\left(  \exp\left[  -i\Omega H_{\rho
}/2\right]  l\right)  \left(  \rho_{0}\right)  ^{\frac{1}{2}}\left(
\sigma_{0}\right)  ^{\beta}\left(  \rho_{0}\right)  ^{\frac{1}{2}%
}\Bigg)^{\alpha}\Bigg\}.
\end{multline}
By applying \eqref{eq: displ gauss 2} of Corollary \ref{cor: displ op} once
again, we obtain
\begin{align}
&  \left(  \frac{1}{Z_{\rho}^{\alpha}Z_{\sigma}^{1-\alpha}}\right)
^{-1}\operatorname{Tr}\left\{  \left(  \rho^{1/2}\sigma^{\left(
1-\alpha\right)  /\alpha}\rho^{1/2}\right)  ^{\alpha}\right\} \nonumber\\
&  =\left(  e^{-\frac{1}{4}l^{T}i\Omega W_{\sigma\left(  \beta\right)  }%
l}e^{\frac{1}{4}\left(  \exp\left[  -i\Omega H_{\rho}/2\right]  l\right)
^{T}i\Omega W_{\xi}\exp\left[  -i\Omega H_{\rho}/2\right]  l}\right)
^{\alpha}\operatorname{Tr}\Bigg\{\Bigg(D(-t)\left(  \rho_{0}\right)
^{\frac{1}{2}}\left(  \sigma_{0}\right)  ^{\beta}\left(  \rho_{0}\right)
^{\frac{1}{2}}D\left(  t\right)  \Bigg)^{\alpha}\Bigg\}\\
&  =\left(  e^{-\frac{1}{4}l^{T}i\Omega W_{\sigma\left(  \beta\right)  }%
l}e^{\frac{1}{4}\left(  \exp\left[  -i\Omega H_{\rho}/2\right]  l\right)
^{T}i\Omega W_{\xi}\exp\left[  -i\Omega H_{\rho}/2\right]  l}\right)
^{\alpha}\operatorname{Tr}\Bigg\{\Bigg(\left(  \rho_{0}\right)  ^{\frac{1}{2}%
}\left(  \sigma_{0}\right)  ^{\beta}\left(  \rho_{0}\right)  ^{\frac{1}{2}%
}\Bigg)^{\alpha}\Bigg\}, \label{eq: Sandwiched RRE imp 1}%
\end{align}
where
\begin{equation}
W_{\xi}=\frac{I+\exp\left[  i\Omega H_{\rho}/2\right]  \exp\left[
i\Omega\beta H_{\sigma}\right]  \exp\left[  i\Omega H_{\rho}/2\right]
}{I-\exp\left[  i\Omega H_{\rho}/2\right]  \exp\left[  i\Omega\beta H_{\sigma
}\right]  \exp\left[  i\Omega H_{\rho}/2\right]  },
\end{equation}
and we have used
\begin{equation}
t=\left(  \exp\left[  -i\Omega H_{\rho}/2\right]  \exp\left[  -i\Omega\beta
H_{\sigma}\right]  \exp\left[  -i\Omega H_{\rho}/2\right]  -I\right)
^{-1}\exp\left[  -i\Omega H_{\rho}/2\right]  \left(  \exp\left[  -i\Omega
H_{\sigma\left(  \beta\right)  }\right]  -I\right)  \left(  -\delta s\right)
.
\end{equation}
Note that the particular value of $t$ is irrelevant because the operators
$D\left(  t\right)  $ and $D(-t)$ cancel each other in the trace operation. We
now simplify the expression in \eqref{eq: Sandwiched RRE imp 1} term by term.
First, consider the exponent in the first prefactor:
\begin{align}
&  -\frac{1}{4}l^{T}i\Omega W_{\sigma\left(  \beta\right)  }l\nonumber\\
&  =-\frac{1}{4}\delta s^{T}i\Omega\left(  \exp\left[  i\Omega H_{\sigma
\left(  \beta\right)  }\right]  -I\right)  \left(  \frac{\exp\left[  -i\Omega
H_{\sigma\left(  \beta\right)  }\right]  +I}{\exp\left[  -i\Omega
H_{\sigma\left(  \beta\right)  }\right]  -I}\right)  \left(  \exp\left[
-i\Omega H_{\sigma\left(  \beta\right)  }\right]  -I\right)  \delta s\\
&  =\frac{1}{4}\delta s^{T}i\Omega\left(  \exp\left[  -i\Omega H_{\sigma
\left(  \beta\right)  }\right]  -\exp\left[  i\Omega H_{\sigma\left(
\beta\right)  }\right]  \right)  \delta s.
\label{eq: sandwiched RRE prefactor 1}%
\end{align}
Second, consider the exponent in the second prefactor in
\eqref{eq: Sandwiched RRE imp 1}:
\begin{align}
&  \frac{1}{4}\left(  \exp\left[  -i\Omega H_{\rho}/2\right]  l\right)
^{T}i\Omega W_{\xi}\exp\left[  -i\Omega H_{\rho}/2\right]  l\nonumber\\
&  =\frac{1}{4}l^{T}i\Omega\exp\left[  i\Omega H_{\rho}/2\right]  W_{\xi}%
\exp\left[  -i\Omega H_{\rho}/2\right]
l\label{eq: sandwiched RRE prefactor 2 beginning}\\
&  =\left(  \left(  \exp\left[  -i\Omega H_{\sigma\left(  \beta\right)
}\right]  -I\right)  \delta s\right)  ^{T}i\Omega\exp\left[  i\Omega H_{\rho
}/2\right]  W_{\xi}\exp\left[  -i\Omega H_{\rho}/2\right]  \left(  \exp\left[
-i\Omega H_{\sigma\left(  \beta\right)  }\right]  -I\right)  \delta s\\
&  =\delta s^{T}\left(  \exp\left[  H_{\sigma\left(  \beta\right)  }%
i\Omega\right]  -I\right)  i\Omega\exp\left[  i\Omega H_{\rho}/2\right]
W_{\xi}\exp\left[  -i\Omega H_{\rho}/2\right]  \left(  \exp\left[  -i\Omega
H_{\sigma\left(  \beta\right)  }\right]  -I\right)  \delta s\\
&  =\delta s^{T}i\Omega\left(  \exp\left[  i\Omega H_{\sigma\left(
\beta\right)  }\right]  -I\right)  \exp\left[  i\Omega H_{\rho}/2\right]
W_{\xi}\exp\left[  -i\Omega H_{\rho}/2\right]  \left(  \exp\left[  -i\Omega
H_{\sigma\left(  \beta\right)  }\right]  -I\right)  \delta s\\
&  =\delta s^{T}i\Omega\left(  \exp\left[  i\Omega H_{\sigma\left(
\beta\right)  }\right]  -I\right)  W_{\xi^{\prime}}\left(  \exp\left[
-i\Omega H_{\sigma\left(  \beta\right)  }\right]  -I\right)  \delta s,
\label{eq: sandwiched RRE prefactor 2}%
\end{align}
where
\begin{equation}
W_{\xi^{\prime}}=\frac{I+\exp\left[  i\Omega H_{\rho}\right]  \exp\left[
i\Omega H_{\sigma\left(  \beta\right)  }\right]  }{I-\exp\left[  i\Omega
H_{\rho}\right]  \exp\left[  i\Omega H_{\sigma\left(  \beta\right)  }\right]
}.
\end{equation}
Based on \eqref{eq: sandwiched RRE prefactor 1} and
\eqref{eq: sandwiched RRE prefactor 2}, and applying Lemma
\ref{lem: displ term} with $H_{9}=H_{\sigma\left(  \beta\right)  }$ and
$H_{10}=H_{\rho}$, we arrive at
\begin{equation}
\left(  e^{-\frac{1}{4}l^{T}i\Omega W_{\sigma\left(  \beta\right)  }l}%
e^{\frac{1}{4}\left(  \exp\left[  -i\Omega H_{\rho}/2\right]  l\right)
^{T}i\Omega W_{\xi}\exp\left[  -i\Omega H_{\rho}/2\right]  l}\right)
^{\alpha}=\exp\left\{  -\alpha\ \delta s^{T}\left(  V_{\sigma\left(
\beta\right)  }+V_{\rho}\right)  ^{-1}\delta s\right\}  .
\end{equation}

Finally, we evaluate the term
\begin{equation}
\operatorname{Tr}\Bigg\{\Bigg(\left(  \rho_{0}\right)  ^{\frac{1}{2}}\left(
\sigma_{0}\right)  ^{\beta}\left(  \rho_{0}\right)  ^{\frac{1}{2}%
}\Bigg)^{\alpha}\Bigg\}.
\end{equation}
By Proposition~\ref{prop:rho-to-alpha}, the covariance matrix $V_{\sigma
(\beta)}$ corresponding to $\beta H_{\sigma}$ is as given in
\eqref{eq:sigma-beta-sandwich-q}. By
Proposition~\ref{prop:sandwich-with-square-root}, we can write
\begin{align}
\left(  \rho_{0}\right)  ^{\frac{1}{2}}\left(  \sigma_{0}\right)  ^{\beta
}\left(  \rho_{0}\right)  ^{\frac{1}{2}}
& =\left[  \exp\left[  -\frac{1}{2}\hat{x}^{T}H_{\rho}\hat{x}\right]  \right]
^{\frac{1}{2}}\exp\left[  -\frac{1}{2}\hat{x}^{T}\beta H_{\sigma}\hat
{x}\right]  \left[  \exp\left[  -\frac{1}{2}\hat{x}^{T}H_{\rho}\hat{x}\right]
\right]  ^{\frac{1}{2}}\\
&=\exp\left[  -\frac{1}{2}\hat{x}^{T}H_{\xi}\hat
{x}\right]  ,
\end{align}
where the covariance matrix corresponding to $H_{\xi}$ is $V_{\xi}$, given in
\eqref{eq:sandwich-intermediate-CM}. Then we have that
\begin{multline}
\operatorname{Tr}\left\{  \left(  \left[  \exp\left[  -\frac{1}{2}\hat{x}%
^{T}H_{\rho}\hat{x}\right]  \right]  ^{\frac{1}{2}}\exp\left[  -\frac{1}%
{2}\hat{x}^{T}\beta H_{\sigma}\hat{x}\right]  \left[  \exp\left[  -\frac{1}%
{2}\hat{x}^{T}H_{\rho}\hat{x}\right]  \right]  ^{\frac{1}{2}}\right)
^{\alpha}\right\} \\
=\operatorname{Tr}\left\{  \left[  \exp\left[  -\frac{1}{2}\hat{x}^{T}H_{\xi
}\hat{x}\right]  \right]  ^{\alpha}\right\}  =\operatorname{Tr}\left\{
\exp\left[  -\frac{1}{2}\hat{x}^{T}\alpha H_{\xi}\hat{x}\right]  \right\}  .
\end{multline}
By Proposition~\ref{prop:rho-to-alpha}, the covariance matrix corresponding to
$\alpha H_{\xi}$ is $V_{\xi(\alpha)}$, given in
\eqref{eq:sandwich-V-ksi-alpha}. This finally implies that
\begin{equation}
\operatorname{Tr}\left\{  \exp\left[  -\frac{1}{2}\hat{x}^{T}\alpha H_{\xi
}\hat{x}\right]  \right\}  =\left[  \det\left(  V_{\xi(\alpha)}+i\Omega
/2\right)  \right]  ^{1/2}.
\end{equation}
Combining the different terms, we arrive at the statement of the theorem.
\end{proof}

\begin{theorem}
\label{thm:sandwiched-alpha>1}
Let $\alpha\in(1,\infty)$, and let 
$\rho$ and $\sigma$ denote 
two Gaussian states such that $V_{\sigma(\gamma)}>V_{\rho}$.
Then the sandwiched R\'{e}nyi relative quasi-entropy
$\widetilde{Q}_{\alpha}(\rho\Vert\sigma)$ as defined in
\eqref{eq:sandwiched-quasi}  is given
by%
\begin{equation}
\widetilde{Q}_{\alpha}(\rho\Vert\sigma)=\frac{Z_{\sigma}^{\alpha-1}}{Z_{\rho
}^{\alpha}}\left[  \det\left(  \left[  V_{\xi(\alpha)}+i\Omega\right]
/2\right)  \right]  ^{1/2}\exp\left\{  \alpha\ \delta s^{T}\left(
V_{\sigma\left(  \gamma\right)  }-V_{\rho}\right)  ^{-1}\delta s\right\}  ,
\label{eq:sandwiched-Renyi-alpha>1}%
\end{equation}
where%
\begin{align}
V_{\xi(\alpha)}  &  =\frac{\left(  I+\left(  V_{\xi}i\Omega\right)
^{-1}\right)  ^{\alpha}+\left(  I-\left(  V_{\xi}i\Omega\right)  ^{-1}\right)
^{\alpha}}{\left(  I+\left(  V_{\xi}i\Omega\right)  ^{-1}\right)  ^{\alpha
}-\left(  I-\left(  V_{\xi}i\Omega\right)  ^{-1}\right)  ^{\alpha}}%
i\Omega,\label{eq:V-ksi-alpha-sandwiched-alpha>1}\\
V_{\xi}  &  =V_{\rho}+\sqrt{I+(V_{\rho}\Omega)^{-2}}V_{\rho}(V_{\sigma
(\gamma)}-V_{\rho})^{-1}V_{\rho}\sqrt{I+(\Omega V_{\rho})^{-2}}%
,\label{eq:V-ksi-sandwiched-alpha>1}\\
V_{\sigma(\gamma)}  &  =\frac{\left(  I+\left(  V_{\sigma}i\Omega\right)
^{-1}\right)  ^{\gamma}+\left(  I-\left(  V_{\sigma}i\Omega\right)
^{-1}\right)  ^{\gamma}}{\left(  I+\left(  V_{\sigma}i\Omega\right)
^{-1}\right)  ^{\gamma}-\left(  I-\left(  V_{\sigma}i\Omega\right)
^{-1}\right)  ^{\gamma}}i\Omega,\label{eq:V-sigma-beta-alpha>1}\\
\gamma &  =\left(  \alpha-1\right)  /\alpha,\\
\delta s  &  =s_{\rho}-s_{\sigma}.
\end{align}

\end{theorem}

\begin{proof}
Let $\rho_{0}$ and $\sigma_{0}$ denote the following operators:
\begin{equation}
\rho_{0}=\exp\left[  -\frac{1}{2}\hat{x}^{T}H_{\rho}\hat{x}\right]
,\qquad\sigma_{0}=\exp\left[  -\frac{1}{2}\hat{x}^{T}H_{\sigma}\hat{x}\right]
.
\end{equation}
Consider that
\begin{align}
&  \operatorname{Tr}\left\{  \left(  \rho^{1/2}\sigma^{\left(  1-\alpha
\right)  /\alpha}\rho^{1/2}\right)  ^{\alpha}\right\} \nonumber\\
&  =\operatorname{Tr}\left\{  \left(  \rho^{1/2}\left[  \sigma^{\left(
\alpha-1\right)  /\alpha}\right]  ^{-1}\rho^{1/2}\right)  ^{\alpha}\right\}
=\operatorname{Tr}\left\{  \left(  \rho^{1/2}\left[  \sigma^{\gamma}\right]
^{-1}\rho^{1/2}\right)  ^{\alpha}\right\} \\
&  =\operatorname{Tr}\left\{  \left(  \left[  D(-s_{\rho})\left(  \frac
{\rho_{0}}{Z_{\rho}}\right)  D(s_{\rho})\right]  ^{\frac{1}{2}}\left[  \left[
D(-s_{\sigma})\left(  \frac{\sigma_{0}}{Z_{\sigma}}\right)  D(s_{\sigma
})\right]  ^{\gamma}\right]  ^{-1}\left[  D(-s_{\rho})\left(  \frac{\rho_{0}%
}{Z_{\rho}}\right)  D(s_{\rho})\right]  ^{\frac{1}{2}}\right)  ^{\alpha
}\right\} \\
&  =\operatorname{Tr}\left\{  \left(  D(-s_{\rho})\left(  \frac{\rho_{0}%
}{Z_{\rho}}\right)  ^{\frac{1}{2}}D\left(  \delta s\right)  \left(  \left[
\frac{\sigma_{0}}{Z_{\sigma}}\right]  ^{\gamma}\right)  ^{-1}D\left(  -\delta
s\right)  \left(  \frac{\rho_{0}}{Z_{\rho}}\right)  ^{\frac{1}{2}}D(s_{\rho
})\right)  ^{\alpha}\right\} \\
&  =\frac{Z_{\sigma}^{\alpha-1}}{Z_{\rho}^{\alpha}}\operatorname{Tr}\left\{
D(-s_{\rho})\left(  \left(  \rho_{0}\right)  ^{\frac{1}{2}}D\left(  \delta
s\right)  \left[  \sigma_{0}^{\gamma}\right]  ^{-1}D\left(  -\delta s\right)
\left(  \rho_{0}\right)  ^{\frac{1}{2}}\right)  ^{\alpha}D(s_{\rho})\right\}
\\
&  =\frac{Z_{\sigma}^{\alpha-1}}{Z_{\rho}^{\alpha}}\operatorname{Tr}\left\{
\left(  \left(  \rho_{0}\right)  ^{\frac{1}{2}}D\left(  \delta s\right)
\left[  \sigma_{0}^{\gamma}\right]  ^{-1}D\left(  -\delta s\right)  \left(
\rho_{0}\right)  ^{\frac{1}{2}}\right)  ^{\alpha}\right\}  .
\label{eq:sandwiched Inverse 1st block last line}%
\end{align}
Using steps similar to those in Theorem~\ref{thm: Sandwiched RRE a<1}, and
based on Remark~\ref{rem: Corollary displ Gaussian inverse}, we arrive at
\begin{multline}
\left(  \frac{Z_{\sigma}^{\alpha-1}}{Z_{\rho}^{\alpha}}\right)  ^{-1}%
\operatorname{Tr}\left\{  \left(  \rho^{1/2}\left[  \sigma^{\gamma}\right]
^{-1}\rho^{1/2}\right)  ^{\alpha}\right\} \\
=\left(  e^{\frac{1}{4}l^{T}i\Omega W_{\sigma\left(  \gamma\right)  }%
l}e^{\frac{1}{4}\left(  \exp\left[  -i\Omega H_{\rho}/2\right]  l\right)
^{T}i\Omega W_{\xi}\exp\left[  -i\Omega H_{\rho}/2\right]  l}\right)
^{\alpha}\operatorname{Tr}\Bigg\{\Bigg(\left(  \rho_{0}\right)  ^{\frac{1}{2}%
}\left[  \sigma^{\gamma}\right]  ^{-1}\left(  \rho_{0}\right)  ^{\frac{1}{2}%
}\Bigg)^{\alpha}\Bigg\} \label{eq: Sandwiched RRE imp 1-1}%
\end{multline}
where%
\begin{align}
W_{\xi}  &  =\frac{I+\exp\left[  i\Omega H_{\rho}/2\right]  \exp\left[
-i\Omega\gamma H_{\sigma}\right]  \exp\left[  i\Omega H_{\rho}/2\right]
}{I-\exp\left[  i\Omega H_{\rho}/2\right]  \exp\left[  -i\Omega\gamma
H_{\sigma}\right]  \exp\left[  i\Omega H_{\rho}/2\right]  },\\
l  &  =\left(  \exp\left[  i\Omega H_{\sigma\left(  \gamma\right)  }\right]
-I\right)  \left(  -\delta s\right)  ,
\end{align}
and where $W_{\sigma\left(  \gamma\right)  }$ is related to $H_{\sigma\left(
\gamma\right)  }=\gamma H_{\sigma}$ by \eqref{eq:H-to-W}. Following steps
similar to those in
\eqref{eq: sandwiched RRE prefactor 2 beginning}--\eqref{eq: sandwiched RRE prefactor 2}
of Theorem~\ref{thm: Sandwiched RRE a<1}, we get
\begin{multline}
\left(  \frac{Z_{\sigma}^{\alpha-1}}{Z_{\rho}^{\alpha}}\right)  ^{-1}%
\operatorname{Tr}\left\{  \left(  \rho^{1/2}\left[  \sigma^{\gamma}\right]
^{-1}\rho^{1/2}\right)  ^{\alpha}\right\} \\
=\left(  e^{\frac{1}{4}l^{T}i\Omega W_{\sigma\left(  \gamma\right)  }%
l}e^{\frac{1}{4}l^{T}i\Omega W_{\xi^{\prime}}l}\right)  ^{\alpha
}\operatorname{Tr}\Bigg\{\Bigg(\rho_{0}^{1/2}\left[  \sigma^{\gamma}\right]
^{-1}\rho_{0}^{1/2}\Bigg)^{\alpha}\Bigg\},
\end{multline}
where
\begin{equation}
W_{\xi^{\prime}}=\frac{I+\exp\left[  i\Omega H_{\rho}\right]  \exp\left[
-i\Omega H_{\sigma\left(  \gamma\right)  }\right]  }{I-\exp\left[  i\Omega
H_{\rho}\right]  \exp\left[  -i\Omega H_{\sigma\left(  \gamma\right)
}\right]  }.
\end{equation}
Applying Lemma~\ref{lem: displ term with inverse} with $H_{9}=H_{\sigma\left(
\gamma\right)  }$ and $H_{10}=H_{\rho}$, we arrive at
\begin{equation}
\left(  e^{\frac{1}{4}l^{T}i\Omega W_{\sigma\left(  \gamma\right)  }l}%
e^{\frac{1}{4}l^{T}i\Omega W_{\xi^{\prime}}l}\right)  ^{\alpha}=\exp\left\{
\alpha\ \delta s^{T}\left(  V_{\sigma\left(  \gamma\right)  }-V_{\rho}\right)
^{-1}\delta s\right\}  . \label{eq: sandwiched renyi a>1 prefactor}%
\end{equation}

Now consider that
\begin{align}
&  \operatorname{Tr}\left\{  \left(  \rho_{0}^{1/2}\left[  \sigma_{0}^{\gamma
}\right]  ^{-1}\rho_{0}^{1/2}\right)  ^{\alpha}\right\} \nonumber\\
&  =\operatorname{Tr}\left\{  \left(  \left[  \exp\left[  -\frac{1}{2}\hat
{x}^{T}H_{\rho}\hat{x}\right]  \right]  ^{\frac{1}{2}}\left[  \left[
\exp\left[  -\frac{1}{2}\hat{x}^{T}H_{\sigma}\hat{x}\right]  \right]
^{\gamma}\right]  ^{-1}\left[  \exp\left[  -\frac{1}{2}\hat{x}^{T}H_{\rho}%
\hat{x}\right]  \right]  ^{\frac{1}{2}}\right)  ^{\alpha}\right\} \\
&  =\operatorname{Tr}\left\{  \left(  \exp\left[  -\frac{1}{2}\hat{x}%
^{T}\left[  H_{\rho}/2\right]  \hat{x}\right]  \exp\left[  -\frac{1}{2}\hat
{x}^{T}\left(  \gamma H_{\sigma}\right)  \hat{x}\right]  ^{-1}\exp\left[
-\frac{1}{2}\hat{x}^{T}\left[  H_{\rho}/2\right]  \hat{x}\right]  \right)
^{\alpha}\right\}  . \label{eq:sandwiched-1st-block-lastline-a>1}%
\end{align}
Let $V_{\sigma(\gamma)}$ denote the covariance matrix corresponding to $\gamma
H_{\sigma}$.\ From Proposition~\ref{prop:rho-to-alpha}, we know that it is
given by \eqref{eq:V-sigma-beta-alpha>1}. By applying
Proposition~\ref{prop:inverse-sandwiched-by-sqrt}, we find that
\begin{equation}
\exp\left[  -\frac{1}{2}\hat{x}^{T}\left[  H_{\rho}/2\right]  \hat{x}\right]
\exp\left[  -\frac{1}{2}\hat{x}^{T}\left(  \gamma H_{\sigma}\right)  \hat
{x}\right]  ^{-1}\exp\left[  -\frac{1}{2}\hat{x}^{T}\left[  H_{\rho}/2\right]
\hat{x}\right]  =\exp\left[  -\frac{1}{2}\hat{x}^{T}H_{\xi}\hat{x}\right]  ,
\end{equation}
where the covariance matrix $V_{\xi}$\ corresponding to $H_{\xi}$ is given by
\eqref{eq:V-ksi-sandwiched-alpha>1}. Furthermore, the covariance matrix
$V_{\xi}$ is legitimate because $H_{\xi}>0$, which in turn follows from the
assumption $V_{\sigma(\gamma)}-V_{\rho}>0$ and the discussion in the proof of
Proposition~\ref{prop:inverse-sandwiched-by-sqrt}. Then we find that
\begin{multline}
\operatorname{Tr}\left\{  \left(  \exp\left[  -\frac{1}{2}\hat{x}^{T}\left[
H_{\rho}/2\right]  \hat{x}\right]  \exp\left[  -\frac{1}{2}\hat{x}^{T}\left(
\gamma H_{\sigma}\right)  \hat{x}\right]  ^{-1}\exp\left[  -\frac{1}{2}\hat
{x}^{T}\left[  H_{\rho}/2\right]  \hat{x}\right]  \right)  ^{\alpha}\right\}
\\
=\operatorname{Tr}\left\{  \left(  \exp\left[  -\frac{1}{2}\hat{x}^{T}H_{\xi
}\hat{x}\right]  \right)  ^{\alpha}\right\}  =\operatorname{Tr}\left\{
\left(  \exp\left[  -\frac{1}{2}\hat{x}^{T}\left[  \alpha H_{\xi}\right]
\hat{x}\right]  \right)  \right\}  .
\end{multline}
By Proposition~\ref{prop:rho-to-alpha}, the covariance matrix $V_{\xi(\alpha
)}$ corresponding to $\alpha H_{\xi}$ is given by
\eqref{eq:V-ksi-alpha-sandwiched-alpha>1}. We can then conclude that
\begin{equation}
\operatorname{Tr}\left\{  \left(  \exp\left[  -\frac{1}{2}\hat{x}^{T}\left[
\alpha H_{\xi}\right]  \hat{x}\right]  \right)  \right\}  =\left[  \det
(V_{\xi(\alpha)}+i\Omega/2)\right]  ^{1/2}.
\end{equation}
This, along with \eqref{eq: sandwiched renyi a>1 prefactor} implies \eqref{eq:sandwiched-Renyi-alpha>1}.
\end{proof}

\bigskip

The collision relative entropy is a special case of the sandwiched R\'enyi
relative entropy, introduced in Ref.~\onlinecite[Definition~5.3.1]{Renner2005} and
applied in subsequent work \cite{BCW14,DFW15,Dupuis2014,BG14}. As a special
case of Theorem~\ref{thm:sandwiched-alpha>1}, we arrive at the following
expression for the collision relative quasi-entropy $\widetilde{Q}_{2}%
(\rho\Vert\sigma)$ after applying Corollaries~\ref{cor:rho-to-the-2} and
\ref{prop:rho-to-1/2}:

\begin{corollary}
The collision relative quasi-entropy $\widetilde{Q}_{2}(\rho\Vert\sigma)$ of
two Gaussian states $\rho$ and $\sigma$ such that $V_{\sigma(1/2)}-V_{\rho}>0$
is given by
\begin{align}
\widetilde{Q}_{2}(\rho\Vert\sigma)  &  =\operatorname{Tr}\left\{  \left(
\rho^{1/2}\sigma^{-1/2}\rho^{1/2}\right)  ^{2}\right\} \\
&  =\frac{Z_{\sigma}}{Z_{\rho}^{2}}\left[  \det\left(  \left[  V_{\xi
(2)}+i\Omega\right]  /2\right)  \right]  ^{1/2}\exp\left\{  2\ \delta
s^{T}\left(  V_{\sigma\left(  1/2\right)  } -V_{\rho} \right)  ^{-1}\delta
s\right\}  ,
\end{align}
where%
\begin{align}
V_{\xi(2)}  &  =\frac{1}{2}(V_{\xi}+\Omega V_{\xi}^{-1}\Omega^{T}),\\
V_{\xi}  &  =V_{\rho}+\sqrt{I+(V_{\rho}\Omega)^{-2}}V_{\rho}(V_{\sigma
(1/2)}-V_{\rho})^{-1}V_{\rho}\sqrt{I+(\Omega V_{\rho})^{-2}},\\
V_{\sigma(1/2)}  &  =\left(  \sqrt{I+(V_{\sigma}\Omega)^{-2}}+I\right)
V_{\sigma}.
\end{align}

\end{corollary}

\section{Max-relative entropy}

\label{sec:max-rel-ent}

Now we derive a formula for the max-relative entropy \cite{D09}, which is
defined as%
\begin{equation}
D_{\max}(\rho\Vert\sigma)\equiv\ln\left\Vert \rho^{1/2}\sigma^{-1}\rho
^{1/2}\right\Vert _{\infty}.
\end{equation}

\begin{theorem}
\label{thm:max-rel-ent}For the case in which $V_{\sigma}-V_{\rho}>0$, the
max-relative entropy $D_{\max}(\rho\Vert\sigma)$ of two Gaussian states $\rho$
and $\sigma$ is given by%
\begin{equation}
D_{\max}(\rho\Vert\sigma)=\ln\left(  \frac{Z_{\sigma}}{Z_{\rho}}\right)
-\sum_{j=1}^{n}\operatorname{arcoth}(\nu_{j}^{\prime})+\delta s^{T}\left(
V_{\sigma}-V_{\rho}\right)  ^{-1}\delta s,
\label{eq:max-rel-ent-symp-to-carry}%
\end{equation}
where $\delta s = s_\rho - s_\sigma$ and $\nu_{j}^{\prime}$ is the $j$th symplectic eigenvalue of the following
covariance matrix:%
\begin{equation}
V^{\prime}=V_{\rho}+\sqrt{I+(V_{\rho}\Omega)^{-2}}V_{\rho}(V_{\sigma}-V_{\rho
})^{-1}V_{\rho}\sqrt{I+(\Omega V_{\rho})^{-2}}.
\end{equation}
Alternatively, we have that%
\begin{equation}
D_{\max}(\rho\Vert\sigma)=\ln\left(  \frac{Z_{\sigma}}{Z_{\rho}}\right)
-\frac{1}{2}\operatorname{Tr}\left\{  \operatorname{arcoth}\left(
\sqrt{-V^{\prime}\Omega V^{\prime}\Omega}\right)  \right\}  +\delta
s^{T}\left(  V_{\sigma}-V_{\rho}\right)  ^{-1}\delta s.
\label{eq:max-rel-ent-func-matrix}%
\end{equation}

\end{theorem}

\begin{proof}
To begin with, we discuss how to calculate the maximum eigenvalue of a
Gaussian state $\omega$ (i.e., $\left\Vert \omega\right\Vert _{\infty}$).
One can also find a discussion of this calculation in Ref.~\onlinecite{Holevo2011}.
Consider that a thermal state $\theta(N)$\ with mean photon number $N\geq0$ is
of the form $\sum_{n=0}^{\infty}N^{n}/\left(N+1\right)^{n+1}|n\rangle\langle n|$, so that
its maximum eigenvalue is equal to $\left[  N+1\right]  ^{-1}$ (corresponding
to the weight of the vacuum). From the Williamson theorem, we know that any
$n$-mode Gaussian state can be written as a unitary operator acting on a
tensor product of $n$ thermal states, and the mean photon number $N$\ of each
thermal state is related to the symplectic eigenvalue $\nu$ as $\nu=2N+1$. In
terms of the symplectic eigenvalue $\nu=2N+1$, the maximum eigenvalue of
$\theta(N)$ is equal to $\left[  N+1\right]  ^{-1}=2/(\nu+1)$. So, for a
general Gaussian state $\omega$, if we have the covariance matrix, we simply
perform a Williamson decomposition, and then we find that%
\begin{equation}
\left\Vert \omega\right\Vert _{\infty}=\prod\limits_{j=1}^{n}2/(\nu_{j}+1).
\label{eq:infin-norm-G-state}%
\end{equation}

Let $\rho_{0}$ and $\sigma_{0}$ denote the following operators:%
\begin{equation}
\rho_{0}=\exp\left[  -\frac{1}{2}\hat{x}^{T}H_{\rho}\hat{x}\right]
,\qquad\sigma_{0}=\exp\left[  -\frac{1}{2}\hat{x}^{T}H_{\sigma}\hat{x}\right]
.
\end{equation}
Consider that
\begin{align}
&  \left\Vert \rho^{1/2}\sigma^{-1}\rho^{1/2}\right\Vert _{\infty}\nonumber\\
&  =\left\Vert \left(  D(-s_{\rho})\left[  \frac{\rho_{0}}{Z_{\rho}}\right]
D(s_{\rho})\right)  ^{\frac{1}{2}}\left(  D(-s_{\sigma})\left[  \frac
{\sigma_{0}}{Z_{\sigma}}\right]  D(s_{\sigma})\right)  ^{-1}\left(
D(-s_{\rho})\left[  \frac{\rho_{0}}{Z_{\rho}}\right]  D(s_{\rho})\right)
^{\frac{1}{2}}\right\Vert _{\infty}\\
&  =\left\Vert D(-s_{\rho})\left[  \frac{\rho_{0}}{Z_{\rho}}\right]
^{\frac{1}{2}}D(s_{\rho})D(-s_{\sigma})\left[  \frac{\sigma_{0}}{Z_{\sigma}%
}\right]  ^{-1}D(s_{\sigma})D(-s_{\rho})\left[  \frac{\rho_{0}}{Z_{\rho}%
}\right]  ^{\frac{1}{2}}D(s_{\rho})\right\Vert _{\infty}\\
&  =\frac{Z_{\sigma}}{Z_{\rho}}\left\Vert \left[  \rho_{0}\right]  ^{\frac
{1}{2}}D\left(  \delta s\right)  \left[  \sigma_{0}\right]  ^{-1}D\left(
-\delta s\right)  \left[  \rho_{0}\right]  ^{\frac{1}{2}}\right\Vert _{\infty
},
\end{align}
where we have used the fact that the infinity norm of an operator is invariant
with respect to unitaries. Note that the operator
\begin{equation}
\left[  \rho_{0}\right]  ^{\frac{1}{2}}D\left(  \delta s\right)  \left[
\sigma_{0}\right]  ^{-1}D\left(  -\delta s\right)  \left[  \rho_{0}\right]
^{\frac{1}{2}}%
\end{equation}
is identical to the operator whose trace is evaluated in
\eqref{eq:sandwiched Inverse 1st block last line} of
Theorem~\ref{thm:sandwiched-alpha>1} when $\gamma$ and $\alpha$ are
independently set to $1$ in that expression. Thus, based on the mean-vector-dependent factor that is derived in
\eqref{eq: sandwiched renyi a>1 prefactor} and the fact that $\Vert A\Vert_\infty = \lim_{p\to \infty} \Vert A \Vert_p$, we have that
\begin{equation}
\left\Vert \rho^{1/2}\sigma^{-1}\rho^{1/2}\right\Vert _{\infty}=\left(
\frac{Z_{\sigma}}{Z_{\rho}}\right)  \exp\left\{  \delta s^{T}\left(
V_{\sigma}-V_{\rho}\right)  ^{-1}\delta s\right\}  \left\Vert \rho_{0}%
^{1/2}\sigma_{0}^{-1}\rho_{0}^{1/2}\right\Vert _{\infty}.
\label{eq: max rel entropy main 1}%
\end{equation}

Now consider that%
\begin{align}
&  \left\Vert \rho_{0}^{1/2}\sigma_{0}^{-1}\rho_{0}^{1/2}\right\Vert _{\infty
}\nonumber\\
&  =\left\Vert \left[  \exp\left[  -\frac{1}{2}\hat{x}^{T}H_{\rho}\hat
{x}\right]  \right]  ^{\frac{1}{2}}\left[  \exp\left[  -\frac{1}{2}\hat{x}%
^{T}H_{\sigma}\hat{x}\right]  \right]  ^{-1}\left[  \exp\left[  -\frac{1}%
{2}\hat{x}^{T}H_{\rho}\hat{x}\right]  \right]  ^{\frac{1}{2}}\right\Vert
_{\infty}\\
&  =\left\Vert \exp\left[  -\frac{1}{2}\hat{x}^{T}\left[  H_{\rho}/2\right]
\hat{x}\right]  \exp\left[  -\frac{1}{2}\hat{x}^{T}\left[  -H_{\sigma}\right]
\hat{x}\right]  \exp\left[  -\frac{1}{2}\hat{x}^{T}\left[  H_{\rho}/2\right]
\hat{x}\right]  \right\Vert _{\infty}.
\end{align}
From Proposition~\ref{prop:inverse-sandwiched-by-sqrt}, we conclude that there
exists an $H^{\prime}$ such that%
\begin{equation}
\exp\left[  -\frac{1}{2}\hat{x}^{T}\left[  H_{\rho}/2\right]  \hat{x}\right]
\exp\left[  -\frac{1}{2}\hat{x}^{T}\left[  -H_{\sigma}\right]  \hat{x}\right]
\exp\left[  -\frac{1}{2}\hat{x}^{T}\left[  H_{\rho}/2\right]  \hat{x}\right]
=\exp\left[  -\frac{1}{2}\hat{x}^{T}H^{\prime}\hat{x}\right]  .
\end{equation}
with corresponding covariance matrix $V^{\prime}$\ given by%
\begin{equation}
V^{\prime}=V_{\rho}+\sqrt{I+(V_{\rho}\Omega)^{-2}}V_{\rho}(V_{\sigma}-V_{\rho
})^{-1}V_{\rho}\sqrt{I+(\Omega V_{\rho})^{-2}}.
\end{equation}
Again applying Proposition~\ref{prop:inverse-sandwiched-by-sqrt}, we find that%
\begin{equation}
\operatorname{Tr}\left\{  \exp\left[  -\frac{1}{2}\hat{x}^{T}H^{\prime}\hat
{x}\right]  \right\}  =\left[  \det\left(  \left[  V^{\prime}+i\Omega\right]
/2\right)  \right]  ^{1/2}.
\end{equation}
Continuing, we find that%
\begin{multline}
\left\Vert \rho_{0}^{1/2}\sigma_{0}^{-1}\rho_{0}^{1/2}\right\Vert _{\infty
}=\left[  \det\left(  \left[  V^{\prime}+i\Omega\right]  /2\right)  \right]
^{1/2}\times\\
\left\Vert \frac{\exp\left[  -\frac{1}{2}\hat{x}^{T}\left[  H_{\rho}/2\right]
\hat{x}\right]  \exp\left[  -\frac{1}{2}\hat{x}^{T}\left[  -H_{\sigma}\right]
\hat{x}\right]  \exp\left[  -\frac{1}{2}\hat{x}^{T}\left[  H_{\rho}/2\right]
\hat{x}\right]  }{\left[  \det\left(  \left[  V^{\prime}+i\Omega\right]
/2\right)  \right]  ^{1/2}}\right\Vert _{\infty}.
\end{multline}
The term inside the infinity norm is a state because $V^{\prime}$ is a
legitimate covariance matrix. Using the expression in
\eqref{eq:infin-norm-G-state}\ for the infinity norm of a Gaussian state, we
find that%
\begin{equation}
\left\Vert \frac{\exp\left[  -\frac{1}{2}\hat{x}^{T}\left[  H_{\rho}/2\right]
\hat{x}\right]  \exp\left[  -\frac{1}{2}\hat{x}^{T}\left[  -H_{\sigma}\right]
\hat{x}\right]  \exp\left[  -\frac{1}{2}\hat{x}^{T}\left[  H_{\rho}/2\right]
\hat{x}\right]  }{\left[  \det\left(  \left[  V^{\prime}+i\Omega\right]
/2\right)  \right]  ^{1/2}}\right\Vert _{\infty}=\prod\limits_{j=1}^{n}%
2/(\nu_{j}^{\prime}+1),
\end{equation}
where $\nu_{j}^{\prime}$ is the $j$th symplectic eigenvalue of $V^{\prime}$.
Using the fact that \cite[Eq.~(2.14)]{MM12}%
\begin{equation}
\left[  \det\left(  \left[  V^{\prime}+i\Omega\right]  /2\right)  \right]
^{1/2}=\prod\limits_{j=1}^{n}\frac{1}{2}\sqrt{(\nu_{j}^{\prime}+1)(\nu
_{j}^{\prime}-1)},
\end{equation}
we find that%
\begin{equation}
\left\Vert \rho_{0}^{1/2}\sigma_{0}^{-1}\rho_{0}^{1/2}\right\Vert _{\infty
}=\left[  \det\left(  \left[  V^{\prime}+i\Omega\right]  /2\right)  \right]
^{1/2}\prod\limits_{j=1}^{n}2/(\nu_{j}^{\prime}+1)=\prod\limits_{j=1}^{n}%
\sqrt{\frac{\nu_{j}^{\prime}-1}{\nu_{j}^{\prime}+1}}.
\end{equation}
Taking a logarithm, we see that%
\begin{equation}
\ln\left\Vert \rho_{0}^{1/2}\sigma_{0}^{-1}\rho_{0}^{1/2}\right\Vert _{\infty
}=\sum_{j=1}^{n}\frac{1}{2}\ln\left(  \frac{\nu_{j}^{\prime}-1}{\nu
_{j}^{\prime}+1}\right)  =-\sum_{j=1}^{n}\frac{1}{2}\ln\left(  \frac{\nu
_{j}^{\prime}+1}{\nu_{j}^{\prime}-1}\right)  =-\sum_{j=1}^{n}%
\operatorname{arcoth}(\nu_{j}^{\prime}).
\end{equation}
Combining with \eqref{eq: max rel entropy main 1} gives \eqref{eq:max-rel-ent-symp-to-carry}.

To arrive at the formula in \eqref{eq:max-rel-ent-func-matrix}, consider for a
covariance matrix $V$ with symplectic diagonalization $S(D\oplus D)S^{T}$,
where $S$ is a symplectic matrix and $D$ is a diagonal matrix of symplectic
eigenvalues, we have that (see Ref.~\onlinecite[Appendix~A]{WTLB16})%
\begin{equation}
Vi\Omega=S(U\otimes I_{n})\left(  [-D]\oplus D\right)  \left[  S(U\otimes
I_{n})\right]  ^{-1}, \label{eq:max-rel-ent-last-steps-1}%
\end{equation}
where $U$ is the following unitary matrix:%
\begin{equation}
U\equiv\frac{1}{\sqrt{2}}%
\begin{bmatrix}
1 & 1\\
i & -i
\end{bmatrix}
.
\end{equation}
From this, we see that%
\begin{equation}
-V\Omega V\Omega=\left(  Vi\Omega\right)  \left(  Vi\Omega\right)  =S(U\otimes
I_{n})\left(  [D^{2}\oplus D^{2}\right)  \left[  S(U\otimes I_{n})\right]
^{-1},
\end{equation}
which implies that%
\begin{equation}
\sqrt{-V\Omega V\Omega}=S(U\otimes I_{n})\left(  D\oplus D\right)  \left[
S(U\otimes I_{n})\right]  ^{-1},
\end{equation}
and in turn that%
\begin{align}
\sum_{j=1}^{n}\operatorname{arcoth}(\nu_{j}^{\prime})  &  =\frac{1}%
{2}\operatorname{Tr}\{\operatorname{arcoth}(\left[  D\oplus D\right]  )\}\\
&  =\frac{1}{2}\operatorname{Tr}\{S(U\otimes I_{n})\operatorname{arcoth}%
(\left[  D\oplus D\right]  )\left[  S(U\otimes I_{n})\right]  ^{-1}\}\\
&  =\frac{1}{2}\operatorname{Tr}\{\operatorname{arcoth}(S(U\otimes
I_{n})\left[  D\oplus D\right]  \left[  S(U\otimes I_{n})\right]  ^{-1})\}\\
&  =\frac{1}{2}\operatorname{Tr}\{\operatorname{arcoth}(\sqrt{-V\Omega
V\Omega})\}. \label{eq:max-rel-ent-last-steps-last}%
\end{align}
This concludes the proof.
\end{proof}

In the above theorem, we provided a formula for the max-relative entropy that
holds whenever $V_{\sigma}-V_{\rho}>0$. The proposition below states that the
condition $V_{\sigma}-V_{\rho}\geq0$ is necessary for the max-relative entropy to be finite
(it still remains open to determine whether the condition $V_{\sigma}-V_{\rho
}>0$\ is necessary and sufficient.)

\begin{proposition}
\label{prop:nec-max-rel-ent}A necessary condition for the max-relative entropy
$D_{\max}(\rho\Vert\sigma)$ of two Gaussian states $\rho$ and $\sigma$ to be
finite is that $V_{\sigma}-V_{\rho}\geq0$.
\end{proposition}

\begin{proof}
Suppose that $D_{\max}(\rho\Vert\sigma)<+\infty$, which implies that there
exists a constant $M$ such that $\rho\leq M\sigma$. Then for all
$u\in\mathbb{R}^{2n}$, we can test the inequality on the displaced vacuum
state $|u\rangle=D(u)|0\rangle$, giving that%
\begin{equation}
\langle u|\rho|u\rangle\leq M\langle u|\sigma|u\rangle.
\label{eq:max-ent-constraint}%
\end{equation}
These expectation values on displaced vacuum states form what is known as the
Husimi Q-function, which for Gaussian states is given by the following
Gaussian form \cite{S17}:%
\begin{equation}
\langle u|\rho|u\rangle=\frac{2^{n}}{\sqrt{\det(V_{\rho}+I)}}\exp\left\{
-\left(  u-s_{\rho}\right)  ^{T}\left(  V_{\rho}+I\right)  ^{-1}\left(
u-s_{\rho}\right)  \right\}  .
\end{equation}
The constraint in \eqref{eq:max-ent-constraint} is then equivalent to%
\begin{equation}
\exp\left\{  -\left(  u-s_{\rho}\right)  ^{T}\left(  V_{\rho}+I\right)
^{-1}\left(  u-s_{\rho}\right)  +\left(  u-s_{\sigma}\right)  ^{T}\left(
V_{\sigma}+I\right)  ^{-1}\left(  u-s_{\sigma}\right)  \right\}  \leq
M\sqrt{\frac{\det(V_{\rho}+I)}{\det(V_{\sigma}+I)}},
\end{equation}
which should be obeyed for all $u\in\mathbb{R}^{2n}$. This is possible only if
$\left(  V_{\rho}+I\right)  ^{-1}\geq\left(  V_{\sigma}+I\right)  ^{-1}$,
which implies that we should have $V_{\sigma}\geq V_{\rho}$.
\end{proof}

\begin{remark}
The development at the end of the proof of Theorem~\ref{thm:max-rel-ent}%
\ extends more generally to any function $f:[1,\infty)\rightarrow
\lbrack0,\infty)$. Given a covariance matrix $V$ with symplectic eigenvalues
$\left\{  \nu_{j}\right\}  _{j=1}^{n}$, then%
\begin{align}
\sum_{j=1}^{n}f(\nu_{j})  &  =\frac{1}{2}\operatorname{Tr}\{f(\sqrt{-V\Omega
V\Omega})\},\\
\prod\limits_{j=1}^{n}f(\nu_{j})  &  =\sqrt{\det\left(  f(\sqrt{-V\Omega
V\Omega})\right)  }.
\end{align}
The first equality follows from a development identical to that in
\eqref{eq:max-rel-ent-last-steps-1}--\eqref{eq:max-rel-ent-last-steps-last}.
The second equality follows because%
\begin{align}
\prod\limits_{j=1}^{n}f(\nu_{j})  &  =\sqrt{\det(f\left(  D\oplus D\right)
)}\\
&  =\sqrt{\det\left(  S(U\otimes I_{n})\right)  \det(f\left(  D\oplus
D\right)  )\det\left(  \left[  S(U\otimes I_{n})\right]  ^{-1}\right)  }\\
&  =\sqrt{\det(S(U\otimes I_{n})f\left(  D\oplus D\right)  \left[  S(U\otimes
I_{n})\right]  ^{-1})}\\
&  =\sqrt{\det\left(  f\left[  S(U\otimes I_{n})\left(  D\oplus D\right)
\left[  S(U\otimes I_{n})\right]  ^{-1}\right]  \right)  }\\
&  =\sqrt{\det\left(  f(\sqrt{-V\Omega V\Omega})\right)  },
\end{align}
with some of the steps following from the development in \eqref{eq:max-rel-ent-last-steps-1}--\eqref{eq:max-rel-ent-last-steps-last}.
\end{remark}

\section{Applications}

\label{sec:apps}

\subsection{Quantum state discrimination and hypothesis testing}

\label{sec:state-disc}

Quantum state discrimination is one of the central problems in quantum
information theory \cite{BK14}. It represents the quantum generalization of
the classical statistical decision-theoretic problem of deciding the
probability distribution corresponding to a random variable, given some
candidate distributions. There is an inherent probability of error associated
with the task, which in the classical case is due to the overlap between the
candidate distributions, and in the quantum case is additionally due to the
non-commutativity of the candidate states. The goal in part is to determine
fundamental bounds on the error probability associated with the discrimination
as dictated by the laws of quantum mechanics. Quantum state discrimination is
important in several areas of quantum information, particularly in quantum
communication and cryptography, where information is encoded in nonorthogonal
quantum states, and optimal decoding requires minimum error discrimination at
the quantum limit. Since continuous-variable physical systems such as the
bosonic field modes of electromagnetic radiation form particularly good
carriers of information in communication scenarios, the discrimination of
Gaussian states is especially important, and has been extensively studied in
the past (see, e.g., Ref.~\onlinecite{I11}).

Binary quantum state discrimination is largely studied in two flavors, namely
with symmetric and asymmetric goals in minimizing the two possible types of
error probabilities in decision. In the symmetric case, the goal is to minimize the average
probability of error in discriminating two quantum states. The optimal
measurement achieving the smallest average error probability was determined in
Refs.~\onlinecite{H69,Helstrom_76} and is known as the Helstrom limit. The Helstrom limit
is a function of the trace distance between the candidate states, which, at
least in the finite-dimensional case, becomes more difficult to calculate as
the dimension of the Hilbert space grows larger \cite{W02}. Furthermore, as
far as we are aware, there is no known simple formula for the trace distance
between two Gaussian states. Thus, upper bounds on the Helstrom limit that are
easier to calculate have been developed. In this regard, the quantum Chernoff
bound \cite{ACMB07,CMMAB08,NS09} serves as a good substitute, and it actually
gives an exact characterization of the exponential decay of the average error
probability in the limit when many copies of the state are available.\ The
quantum Chernoff bound can be expressed as an optimized Petz--R\'{e}nyi
relative entropy for $\alpha\in(0,1)$.

In one variant of asymmetric hypothesis testing, the error probability
corresponding to one of the types of errors is constrained to decay at a rate
$e^{-nr}$, for some $r>0$ and where $n$ is the number of copies of the state,
while the goal is to determine the behavior of the other kind of error
probability. If $r$ is less than the quantum relative entropy, the quantum
Hoeffding bound \cite{PhysRevA.76.062301,N06} applies and states that the
other kind of error probability decays exponentially fast to zero, and the
optimal error exponent can be expressed in terms of the Petz--R\'{e}nyi
relative entropy~\cite{P86}. If $r$ exceeds the quantum relative entropy, the
strong converse bound from Ref.~\onlinecite{MO13} applies and states that the other kind
of error probability converges exponentially fast to one, and the optimal
strong converse exponent can be expressed in terms of the sandwiched R\'{e}nyi
relative entropy \cite{MDSFT13,WWY14}.

Gaussian state discrimination has been studied in the contexts of both
symmetric and asymmetric cost of errors. The quantum Chernoff bound
\cite{CMMAB08,PL08} and the quantum Hoeffding bound \cite{SB14} for Gaussian
states have been considered. However, the expressions given in these earlier
works were in terms of the symplectic decomposition of the Gaussian states. A
quest for more compact and elegant expressions for the quantities that solely
depend on the covariance matrices of the candidates states has prompted the
development of other less tight bounds for these quantities \cite{PL08}.

The formulas derived in our paper readily apply to the settings of the quantum
Chernoff bound, the quantum Hoeffding bound, and the strong converse regime
and lead to expressions for the exponents of Gaussian state discrimination in
these contexts.\ We do not give details here, but instead we simply note that
the results can be obtained by direct substitution of our formulas into the
general expressions for the various bounds.
We note that our formulas depend only on the mean vectors and covariance matrices of the candidate states.

\subsection{Quantum communication theory}

There is an intimate link between hypothesis testing and communication theory,
first realized in the classical case in Ref.~\onlinecite{B74}.\ This approach has since
been successfully explored in the context of quantum communication theory
\cite{ON99,PhysRevA.76.062301,KW09,SW12,WWY14,GW13,CMW14,TWW14,DingW15,WTB16},
in order to establish strong converse bounds for a variety of
information-processing tasks. In all of the aforementioned works, the strong
converse bounds are expressed in terms of the R\'{e}nyi relative entropies. As
such, one would expect the formulas derived here to apply in these contexts,
and we now comment on the most direct application of our results in the
context of quantum and private communication.

To begin with, let us recall that a quantum channel has a capacity for quantum
and private communication when assisted by classical communication between the
sender and receiver (see, e.g., Refs.~\onlinecite{BDSW96,BDSS06,TGW14Nat} for these notions).
These capacities are roughly and respectively defined to be the maximum rates
at which these communication resources can be used to establish entanglement
or secret key reliably between a sender and a receiver, when using the channel
many times. It is of interest to understand these capacities in the context of
quantum key distribution \cite{SBCDLP09}, in order to understand the
limitations on practical protocols. For channels that are teleportation
simulable \cite{BDSW96}, meaning that they can be realized by the action of
local operations and classical communication (LOCC)\ on a resource state \cite{HHH99}, a
general protocol of the above form can be significantly simplified
\cite{BDSW96,Mul12}, such that it consists of a single round of LOCC\ acting
on a given number of copies of the resource state. As observed in
Ref.~\onlinecite{BDSW96} for the case of quantum communication, one can then bound the
assisted quantum capacity of the channel in terms of the distillable
entanglement of the resource state, and the same reasoning trivially extends
to the case of assisted private communication. These observations apply as
well to Gaussian channels that are teleportation simulable, as identified and
discussed in Refs.~\onlinecite{WPG07,NFC09}.

One of the main contributions of Refs.~\onlinecite{TWW14,WTB16} is that bounds on the
strong converse exponent for assisted quantum and private communication over
teleportation-simulable channels, respectively, can be expressed in terms of
the sandwiched R\'{e}nyi relative entropy of the underlying resource state.
After these developments, a recent work \cite{KW17}, following the approach of
Ref.~\onlinecite{LMGA17},\ found finite-energy Gaussian resource states that can be used
for the teleportation simulation of thermal Gaussian channels, and as such,
they were used to establish bounds on the assisted quantum and private
capacities of these channels. Avoiding details, we simply note here that one
can evaluate the sandwiched R\'{e}nyi relative entropy of the finite-energy
Gaussian resource states from Ref.~\onlinecite{KW17}\ in order to determine bounds on the
strong converse exponent for  communication over these channels.

\subsection{Mixing times of Markov processes and covert communication}

\label{sec:covert-app}We finally briefly mention some applications of the
Petz--R\'{e}nyi relative entropy of order two. One particular quantum
$\chi^{2}$ divergence from Ref.~\onlinecite{TKRWF10}\ can be related to the
Petz--R\'{e}nyi relative entropy of order two. Therein, the authors used the
quantum $\chi^{2}$ divergence to bound mixing times of quantum Markov
processes. As a result, we suspect that the formulas derived in our paper will
be useful in the context of bounding mixing times of quantum Gauss--Markov
processes, such as the processes considered in Refs.~\onlinecite{HHW09,GHLM10}.

Additionally, the Petz--R\'{e}nyi relative entropy of order two has been
employed in the context of bounding error probabilities for covert communication over
quantum channels \cite{SBTGG16}. In covert communication, the goal is for two
parties to communicate information over a quantum channel, such that someone
else (typically called a warden),\ who is allowed to observe the channel, is
effectively not able to realize that they are in fact communicating. In light
of this previous work, we expect that the formula derived in our paper will be
useful in the context of covert communication when using a quantum Gaussian
channel for the task.

\section{Conclusion}

\label{sec:conclusion}The main contribution of our paper is the derivation of
formulas for the Petz--R\'{e}nyi relative entropy and the sandwiched R\'{e}nyi
relative entropy of quantum Gaussian states for $\alpha\in(0,1)\cup(1,\infty
)$. Interestingly, our approach handles the previously elusive case for the
Petz--R\'{e}nyi relative entropy when $\alpha\in(1,\infty)$. We also derived a
formula for the max-relative entropy of two quantum Gaussian states. Given the
wide applicability of the R\'{e}nyi relative entropies and quantum Gaussian
states in quantum information theory and beyond, we suspect that the formulas
derived here will be useful in a number of future applications.

For future work, it remains open to determine whether the sufficient
conditions given in Theorems~\ref{prop:petz-renyi-alpha>1} and
\ref{thm:sandwiched-alpha>1} are in fact necessary for the quantities to be
finite. The similarity of the sufficient conditions with the necessary and
sufficient conditions from the classical case \cite{G11,GAL13}\ suggest that
this might be the case. At the least, Proposition~\ref{prop:nec-max-rel-ent}%
\ establishes significant progress on this question for the max-relative
entropy. Additionally, the approach given in our paper can be used to
determine expressions for the $\alpha$-$z$ relative entropies \cite{AD15}\ and
the generalized R\'{e}nyi quantities from Refs.~\onlinecite{BSW14,SBW14,BSW15a,DW15} (in
the latter case, we would need expressions for the adjoint of a quantum
Gaussian channel, as given in Ref.~\onlinecite{GLS16}).

\appendix

\section{Covariance matrix for $\rho(\alpha)$}

\label{app:CM-for-rho-alpha}Recall that the covariance matrix $V_{\rho}$ for
an $n$-mode state has a symplectic (Williamson)\ decomposition as%
\begin{equation}
S_{\rho}\left(  D_{\rho}\oplus D_{\rho}\right)  S_{\rho}^{T}=S_{\rho}\left(
I_{2}\otimes D_{\rho}\right)  S_{\rho}^{T},
\end{equation}
where $S_{\rho}$ is a $2n\times2n$ symplectic matrix satisfying $S\Omega
S^{T}=\Omega$ and $D_{\rho}$ is a diagonal matrix of symplectic eigenvalues
(each entry being $>1$ for a faithful state).

\begin{proposition}
The following equality holds%
\begin{align}
V_{\rho(\alpha)}  &  =\frac{\left(  I+\left(  V_{\rho}i\Omega\right)
^{-1}\right)  ^{\alpha}+\left(  I-\left(  V_{\rho}i\Omega\right)
^{-1}\right)  ^{\alpha}}{\left(  I+\left(  V_{\rho}i\Omega\right)
^{-1}\right)  ^{\alpha}-\left(  I-\left(  V_{\rho}i\Omega\right)
^{-1}\right)  ^{\alpha}}i\Omega\\
&  =S_{\rho}\left(  I_{2}\otimes\frac{\left(  D_{\rho}+I\right)  ^{\alpha
}+\left(  D_{\rho}-I\right)  ^{\alpha}}{\left(  D_{\rho}+I\right)  ^{\alpha
}-\left(  D_{\rho}-I\right)  ^{\alpha}}\right)  S_{\rho}^{T},
\end{align}
which demonstrates the equivalence of $V_{\rho(\alpha)}$ with Eqs.~(54)\ and
(55)\ of \cite{PL08}.
\end{proposition}

\begin{proof}
By definition,
\begin{equation}
V_{\rho(\alpha)}=\frac{\left(  I+\left(  V_{\rho}i\Omega\right)  ^{-1}\right)
^{\alpha}+\left(  I-\left(  V_{\rho}i\Omega\right)  ^{-1}\right)  ^{\alpha}%
}{\left(  I+\left(  V_{\rho}i\Omega\right)  ^{-1}\right)  ^{\alpha}-\left(
I-\left(  V_{\rho}i\Omega\right)  ^{-1}\right)  ^{\alpha}}i\Omega.
\end{equation}
Consider the following reasoning along the lines from Ref.~\onlinecite[Appendix~A]%
{WTLB16}. The covariance matrix $V_{\rho}$ for an $n$-mode state has a
symplectic decomposition as%
\begin{equation}
S_{\rho}\left(  D_{\rho}\oplus D_{\rho}\right)  S_{\rho}^{T}=S_{\rho}\left(
I_{2}\otimes D_{\rho}\right)  S_{\rho}^{T},
\end{equation}
where $S_{\rho}$ is a $2n\times2n$ symplectic matrix and $D_{\rho}$ is a
diagonal matrix of symplectic eigenvalues. After some steps, this implies that%
\begin{align}
V_{\rho}i\Omega &  =S_{\rho}\left(  U\otimes I_{n}\right)  \left(  \left[
-D_{\rho}\right]  \oplus D_{\rho}\right)  \left(  U^{\dag}\otimes
I_{n}\right)  S_{\rho}^{-1}\\
&  =S_{\rho}\left(  U\otimes I_{n}\right)  \left(  -\sigma_{Z}\otimes D_{\rho
}\right)  \left(  U^{\dag}\otimes I_{n}\right)  S_{\rho}^{-1},
\end{align}
where%
\begin{equation}
U\equiv\frac{1}{\sqrt{2}}%
\begin{bmatrix}
1 & 1\\
i & -i
\end{bmatrix}
.
\end{equation}
So we see that the eigenvalues of $V_{\rho}i\Omega$ correspond to the
symplectic eigenvalues of $V_{\rho}$ and the eigenvectors of $V_{\rho}i\Omega$
correspond to the symplectic eigenvectors of $V_{\rho}$. Let us abbreviate
this as%
\begin{align}
V_{\rho}i\Omega &  =M\overline{D}M^{-1},\\
M  &  =S_{\rho}\left(  U\otimes I_{n}\right)  ,\\
\overline{D}  &  =\left[  -D_{\rho}\right]  \oplus D_{\rho}.
\end{align}
Note that for a positive definite state $\rho$, each entry of $D$ is $>1$. So
this means that both $Vi\Omega+I$ and $Vi\Omega-I$ are invertible matrices.
Consider that%
\begin{align}
V_{\rho(\alpha)}  &  =\frac{\left(  I+\left(  V_{\rho}i\Omega\right)
^{-1}\right)  ^{\alpha}+\left(  I-\left(  V_{\rho}i\Omega\right)
^{-1}\right)  ^{\alpha}}{\left(  I+\left(  V_{\rho}i\Omega\right)
^{-1}\right)  ^{\alpha}-\left(  I-\left(  V_{\rho}i\Omega\right)
^{-1}\right)  ^{\alpha}}i\Omega\\
&  =\frac{\left(  I+\left(  M\overline{D}M^{-1}\right)  ^{-1}\right)
^{\alpha}+\left(  I-\left(  M\overline{D}M^{-1}\right)  ^{-1}\right)
^{\alpha}}{\left(  I+\left(  M\overline{D}M^{-1}\right)  ^{-1}\right)
^{\alpha}-\left(  I-\left(  M\overline{D}M^{-1}\right)  ^{-1}\right)
^{\alpha}}i\Omega\\
&  =\frac{\left(  M\left(  I+\overline{D}^{-1}\right)  M^{-1}\right)
^{\alpha}+\left(  M\left(  I-\overline{D}^{-1}\right)  M^{-1}\right)
^{\alpha}}{\left(  M\left(  I+\overline{D}^{-1}\right)  M^{-1}\right)
^{\alpha}-\left(  M\left(  I-\overline{D}^{-1}\right)  M^{-1}\right)
^{\alpha}}i\Omega\\
&  =M\frac{\left(  I+\overline{D}^{-1}\right)  ^{\alpha}+\left(
I-\overline{D}^{-1}\right)  ^{\alpha}}{\left(  I+\overline{D}^{-1}\right)
^{\alpha}-\left(  I-\overline{D}^{-1}\right)  ^{\alpha}}M^{-1}i\Omega.
\end{align}
Finally consider that%
\begin{equation}
M^{-1}i\Omega=\left(  U^{\dag}\otimes I_{n}\right)  S_{\rho}^{-1}%
i\Omega=\left(  U^{\dag}\otimes I_{n}\right)  i\Omega S_{\rho}^{T}.
\end{equation}
This then implies that%
\begin{equation}
V_{\rho(\alpha)}=S_{\rho}\left(  U\otimes I_{n}\right)
\frac{\left(  I+\overline{D}^{-1}\right)  ^{\alpha}+\left(  I-\overline
{D}^{-1}\right)  ^{\alpha}}{\left(  I+\overline{D}^{-1}\right)  ^{\alpha
}-\left(  I-\overline{D}^{-1}\right)  ^{\alpha}}\left(  U^{\dag}\otimes
I_{n}\right)  i\Omega S_{\rho}^{T}.
\end{equation}
Since the function $x\rightarrow\frac{\left(  1+1/x\right)  ^{\alpha}+\left(
1-1/x\right)  ^{\alpha}}{\left(  1+1/x\right)  ^{\alpha}-\left(  1-1/x\right)
^{\alpha}}$ is an odd function (it is a composition of three odd functions:
$\operatorname{arcoth}$, scaling by $\alpha$, and then $\coth$) and
$\frac{\left(  1+1/x\right)  ^{\alpha}+\left(  1-1/x\right)  ^{\alpha}%
}{\left(  1+1/x\right)  ^{\alpha}-\left(  1-1/x\right)  ^{\alpha}}%
=\frac{\left(  x+1\right)  ^{\alpha}+\left(  x-1\right)  ^{\alpha}}{\left(
x+1\right)  ^{\alpha}-\left(  x-1\right)  ^{\alpha}}$ for $x>1$, we can
rewrite%
\begin{align}
\frac{\left(  I+\overline{D}^{-1}\right)  ^{\alpha}+\left(  I-\overline
{D}^{-1}\right)  ^{\alpha}}{\left(  I+\overline{D}^{-1}\right)  ^{\alpha
}-\left(  I-\overline{D}^{-1}\right)  ^{\alpha}}  &  =\overline{\frac{\left(
I+D^{-1}\right)  ^{\alpha}+\left(  I-D^{-1}\right)  ^{\alpha}}{\left(
I+D^{-1}\right)  ^{\alpha}-\left(  I-D^{-1}\right)  ^{\alpha}}}\\
&  =-\sigma_{Z}\otimes\frac{\left(  I+D^{-1}\right)  ^{\alpha}+\left(
I-D^{-1}\right)  ^{\alpha}}{\left(  I+D^{-1}\right)  ^{\alpha}-\left(
I-D^{-1}\right)  ^{\alpha}}\\
&  =-\sigma_{Z}\otimes\frac{\left(  D+I\right)  ^{\alpha}+\left(  D-I\right)
^{\alpha}}{\left(  D+I\right)  ^{\alpha}-\left(  D-I\right)  ^{\alpha}},
\end{align}
which means that%
\begin{align}
&  S_{\rho}\left(  U\otimes I_{n}\right)  \frac{\left(  I+\overline{D}%
^{-1}\right)  ^{\alpha}+\left(  I-\overline{D}^{-1}\right)  ^{\alpha}}{\left(
I+\overline{D}^{-1}\right)  ^{\alpha}-\left(  I-\overline{D}^{-1}\right)
^{\alpha}}\left(  U^{\dag}\otimes I_{n}\right)  i\Omega S_{\rho}%
^{T}\nonumber\\
&  =S_{\rho}\left(  U\otimes I_{n}\right)  \left(  -\sigma_{Z}\otimes
\frac{\left(  D+I\right)  ^{\alpha}+\left(  D-I\right)  ^{\alpha}}{\left(
D+I\right)  ^{\alpha}-\left(  D-I\right)  ^{\alpha}}\right)  \left(  U^{\dag
}\otimes I_{n}\right)  i\Omega S_{\rho}^{T}\\
&  =S_{\rho}\left(  -U\sigma_{Z}U^{\dag}\otimes\frac{\left(  D+I\right)
^{\alpha}+\left(  D-I\right)  ^{\alpha}}{\left(  D+I\right)  ^{\alpha}-\left(
D-I\right)  ^{\alpha}}\right)  i\Omega S_{\rho}^{T}\\
&  =S_{\rho}\left(  -\sigma_{Y}\otimes\frac{\left(  D+I\right)  ^{\alpha
}+\left(  D-I\right)  ^{\alpha}}{\left(  D+I\right)  ^{\alpha}-\left(
D-I\right)  ^{\alpha}}\right)  i\Omega S_{\rho}^{T}\\
&  =S_{\rho}\left(  -\sigma_{Y}\otimes\frac{\left(  D+I\right)  ^{\alpha
}+\left(  D-I\right)  ^{\alpha}}{\left(  D+I\right)  ^{\alpha}-\left(
D-I\right)  ^{\alpha}}\right)  \left(  -\sigma_{Y}\otimes I_{n}\right)
S_{\rho}^{T}\\
&  =S_{\rho}\left(  I_{2}\otimes\frac{\left(  D+I\right)  ^{\alpha
}+\left(  D-I\right)  ^{\alpha}}{\left(  D+I\right)  ^{\alpha}-\left(
D-I\right)  ^{\alpha}}\right)  S_{\rho}^{T}.
\end{align}
So we conclude that
\begin{equation}
V_{\rho(\alpha)}=S_{\rho}\left(  I_{2}\otimes\frac{\left(  D+I\right)
^{\alpha}+\left(  D-I\right)  ^{\alpha}}{\left(  D+I\right)  ^{\alpha}-\left(
D-I\right)  ^{\alpha}}\right)  S_{\rho}^{T}.
\end{equation}
This completes the proof of the equivalence of $V_{\rho(\alpha)}$ with
Eqs.~(54)\ and (55)\ of Ref.~\onlinecite{PL08}.
\end{proof}

\bigskip

\textbf{Acknowledgements.} We are grateful to Leonardo Banchi, Saikat Guha,
Felix Leditzky, Ty Volkoff, and Haoyu Qi for discussions.
KPS thanks the Max Planck Society for support.
LL acknowledges financial support from the European Research Council (AdG IRQUAT No.~267386), the Spanish MINECO (project no.~FIS2013-40627-P and no.~FIS2016-86681-P), and the Generalitat de Catalunya (CIRIT Project no.~2014 SGR 966).
MMW\ acknowledges support from
the National Science Foundation and the Office of Naval Research. He is also
grateful to the Max Planck Institute for the Science of Light, especially
Christoph Marquardt, for hosting him for a research visit in June 2016, when
this project was initiated.

\bibliographystyle{unsrt}
\bibliography{Ref}

\end{document}